\documentclass[letterpaper,twocolumn,11pt,accepted=2023-04-22]{quantumarticle}
\pdfoutput=1
\usepackage[utf8]{inputenc}
\usepackage[english]{babel}
\usepackage[T1]{fontenc}

\usepackage[colorlinks=true,breaklinks=true,linkcolor=red,citecolor=magenta,urlcolor=magenta]{hyperref}

\usepackage{graphicx}
\usepackage{dcolumn}
\usepackage{bm}
\usepackage{dsfont}
\usepackage{amsthm}
\usepackage{chngcntr}
\usepackage{apptools}
\usepackage{amsmath}
\usepackage{mathtools}
\usepackage{todonotes}
\usepackage{graphicx}
\graphicspath{ {./images/} }
\usepackage{siunitx}
\usepackage[margin=1in,footskip=0.25in]{geometry}
\usepackage{blindtext}
\usepackage{enumitem}
\usepackage{xcolor}
\usepackage[makeroom]{cancel}

\usepackage{placeins}
\AtAppendix{\counterwithin{lemma}{section}}
\usepackage{longtable} 

\usepackage{algorithm}
\usepackage{algpseudocode}

\usepackage[capitalise]{cleveref}
\crefname{figure}{Fig.}{Fig.}

\newtheorem{theorem}{Theorem}

\newtheorem{lemma}{Lemma}
\theoremstyle{definition}
\newtheorem{definition}{Definition}
\theoremstyle{theorem}

\theoremstyle{definition}

\usepackage{tikz}
\usetikzlibrary{quantikz}

\begin{document}

\title{Exact and efficient Lanczos method on a quantum computer}

\author{William Kirby}
\email{william.kirby@ibm.com}
\affiliation{
IBM Quantum, IBM Research Cambridge, Cambridge, MA 02142, USA
}
\affiliation{
Department of Physics and Astronomy, Tufts University, Medford, MA 02155, USA
}

\author{Mario Motta}
\affiliation{
IBM Quantum, IBM Research Almaden, San Jose, CA 95120, USA
}

\author{Antonio Mezzacapo}
\affiliation{
IBM Quantum, T. J. Watson Research Center, Yorktown Heights, NY 10598, USA
}

\begin{abstract}

We present an algorithm that uses block encoding on a quantum computer to exactly construct a Krylov space, which can be used as the basis for the Lanczos method to estimate extremal eigenvalues of Hamiltonians.
While the classical Lanczos method has exponential cost in the system size to represent the Krylov states for quantum systems, our efficient quantum algorithm achieves this in polynomial time and memory.
The construction presented is exact in the sense that the resulting Krylov space is identical to that of the Lanczos method, so the only approximation with respect to the exact method is due to finite sample noise.
This is possible because, unlike previous quantum Krylov methods, our algorithm does not require simulating real or imaginary time evolution.
We provide an explicit error bound for the resulting ground state energy estimate in the presence of noise.
For our method to be successful efficiently, the only requirement on the input problem is that the overlap of the initial state with the true ground state must be $\Omega(1/\text{poly}(n))$ for $n$ qubits.

\end{abstract}

\maketitle

\section{Introduction}

Estimating the ground state energy of a quantum system has broad applications across physics and chemistry.
In the most challenging cases, classically accomplishing this is believed to be exponentially hard in the system size, due to the cost of classically representing vectors in the Hilbert space~\cite{aspuruguzik2005molecularenergies,peruzzo2014vqe,omalley2016scalable,kandala2017hardwareefficient,mcardle2020quantumcomputational,su2021firstquantization}.
Many quantum algorithms for the task have been proposed (see for example~\cite{aspuruguzik2005molecularenergies,peruzzo2014vqe,omalley2016scalable,kandala2017hardwareefficient,mcardle2020quantumcomputational,su2021firstquantization,kitaev1995phaseestimation,lloyd1996quantumsimulators,farhi2000adiabatic,somma2002physicalphenomena,whitfield2011quantumsimulation,babbush2018lowdepth,mcclean2017subspace,colless2018computation,huggins2020nonorthogonal,motta2020qite_qlanczos,parrish2019filterdiagonalization,stair2020krylov,takeshita2020subspace,lin2020nearoptimalground,cohn2021filterdiagonalization,yoshioka2021virtualsubspace,klymko2022realtime,seki2021powermethod,cortes2022krylov,baek2022nonorthogonal,tkachenko2022davidson,lin2022heisenberglimited,dong2022groundstate}) that avoid this representation problem by encoding the system's Hilbert space in the Hilbert space of a register of qubits, whose dimension also grows exponentially.

The problem of finding ground state energies of local Hamiltonians is QMA-hard~\cite{kitaev02,kempe2006local}.
This means that any sufficiently general quantum algorithms for approximating ground state energies must require additional assumptions.
Some typical additional assumptions are finding an initial guess state that has $\Omega(1/\text{poly}(n))$ overlap with sufficiently low-lying energy states (for $n$ qubits)~\cite{kitaev1995phaseestimation,lin2020nearoptimalground,lin2022heisenberglimited,dong2022groundstate}, finding a gapped adiabatic path from an efficiently solvable Hamiltonian to the target Hamiltonian~\cite{farhi2000adiabatic}, optimizing a variational quantum circuit~\cite{peruzzo2014vqe}, or assuming bounded correlation length in the ground state~\cite{motta2020qite_qlanczos}.

One design principle for quantum algorithms is to reduce the overhead of classical numerical techniques.
This can be applied to classical algorithms that target ground states.
A well-known classical technique for finding extremal eigenvalues is the \emph{Lanczos method}~\cite{lanczos1950iteration,cullum2002lanczos}, which constructs a linear \emph{Krylov space} as the span of states obtained by applying powers of the Hamiltonian to some reference state.
The Hamiltonian is then diagonalized within the Krylov space, obtaining a variational estimate of the ground state energy.
This method has been widely studied and applied, and for a given problem instance it converges exponentially quickly in the dimension of the Krylov space, in practice often reaching sufficient accuracy within only some tens of powers of the Hamiltonian~\cite{kaniel1966linearalgebra,paige1971computation,saad1980lanczos,epperly2021subspacediagonalization}.
However, for quantum systems, the Krylov basis vectors have exponential dimension in the system size, which makes the classical Lanczos method exponentially costly in general.

A natural question is therefore whether a quantum version of the algorithm can make this cost polynomial.
Applying a power of the Hamiltonian operator to a reference state is a nonunitary operation.
Hence, some quantum versions of the Lanczos method have been proposed that replace powers of the Hamiltonian with unitary operators that approximate real~\cite{parrish2019filterdiagonalization,stair2020krylov,cohn2021filterdiagonalization,cortes2022krylov,klymko2022realtime} or imaginary~\cite{motta2020qite_qlanczos} time evolutions generated by the Hamiltonian.
Another strategy approximates powers of the Hamiltonian itself using linear combinations of time evolutions~\cite{seki2021powermethod}.
These methods avoid the exponential cost of representing quantum states, which have dimension $2^n$ for $n$ qubits.
However, they only approximate the original Lanczos method due to the modification of the Krylov basis vectors.
Furthermore, real and imaginary time evolutions can only be approximated on a digital quantum computer.

In this paper, we propose a quantum algorithm that exactly reproduces the Krylov space used in the classical Lanczos method up to finite sample noise, without the exponential classical representation cost, thus avoiding the aforementioned approximation errors.
Unlike in typical classical implementations, our method does not involve iteratively orthogonalizing the Krylov vectors, but it does yield the same Krylov space.
The algorithm we propose is based on the observation that the simplest case of qubitization presented in~\cite{low2019qubitization} applies to a Hamiltonian expressed as a linear combination of Pauli operators and produces block encodings of Chebyshev polynomials of the Hamiltonian. 
It does not require quantum signal processing~\cite{low2017quantumsignalprocessing}, which would be required to approximate time evolution using qubitization.
Instead, we show that applications of the block-encoding unitaries followed by measurement, repeated in order to estimate a collection of expectation values, provide enough information for the classical computer to do the rest.

We also provide an explicit analysis of the ground state energy error from our algorithm in the presence of noise.
The resulting Krylov space dimension (which is proportional to the maximum circuit depth) required to achieve a given target error scales either as the inverse of the target error or of the spectral gap, whichever is more favorable, times logarithmic factors: see \eqref{krylov_dimension_required}.
This means that our method can tolerate a vanishing spectral gap, but also that when the spectral gap is nonzero, the total number of measurements required to achieve a given target error scales asymptotically as the inverse-squared error (times logarithmic factors): see \eqref{measurements_lower_bound}.
This is the scaling one would expect in an algorithm that is based on repeated sampling.

Compared to prior work, if we first compare to previous quantum subspace diagonalization methods~\cite{motta2020qite_qlanczos,parrish2019filterdiagonalization,stair2020krylov,cohn2021filterdiagonalization,seki2021powermethod,cortes2022krylov,klymko2022realtime}, unlike all of these our algorithm is exactly equivalent to the classical Lanczos method up to sampling noise, and does not require approximate real or imaginary time evolution.
Unlike imaginary time-evolution, our algorithm does not require finite correlation length~\cite{motta2020qite_qlanczos}.
Compared to adiabatic state preparation~\cite{farhi2000adiabatic}, our algorithm does not depend on the spectral gap unless that is larger than the target error, and unlike variational quantum eigensolvers~\cite{peruzzo2014vqe}, it does not require optimizing a parameterized quantum circuit.
Finally, compared to quantum phase estimation, our algorithm requires a number of shots that is inverse quadratic in the target error rather than merely inverse (which is preferable).
However, it has shorter circuits if the phase estimation algorithm is also based on block encoding, and our method is robust to noise.
Details of all of these comparisons are given in \cref{conclusion}.

The paper is organized as follows. In \cref{background} we give requisite background on block encoding, qubitization, and quantum subspace diagonalization.
In \cref{method}, we describe and analyze our quantum algorithm for producing a Krylov space.
In \cref{error_analysis}, we describe the classical post-processing (regularization) required for our method, provide an error bound for the resulting ground state energy estimate in the presence of noise, and discuss the resulting scalings of Krylov space dimension and measurements.
In \cref{numerics}, we provide some numerical demonstrations of our method.
Finally, we discuss our results and conclude in \cref{conclusion}.

\section{Background}
\label{background}

\subsection{Block encoding}
\label{block_encoding_sec}

\begin{definition}[block encoding~\cite{low2019qubitization}]
\label{block_encoding}
    A \emph{block encoding} of a Hamiltonian $H$ (acting on a Hilbert space $\mathcal{H}_s$ whose states are denoted $|\cdot\rangle_s$) is a pair $(U,G)$, where $U$ is a unitary acting on $\mathcal{H}_a\otimes\mathcal{H}_s$ (for some auxiliary Hilbert space $\mathcal{H}_a$ whose states are denoted $|\cdot\rangle_a$), and $|G\rangle_a\coloneqq G|0\rangle_a$ is a state of the auxiliary qubits such that
    \begin{equation}
    \label{block_encoding_def}
        (\langle G|_a\otimes\mathds{1}_s)U(|G\rangle_a\otimes\mathds{1}_s)=H.
    \end{equation}
    Here $\mathds{1}_s$ denotes identity acting on $\mathcal{H}_s$.
\end{definition}
\noindent

As an example of block encoding, consider an $n$-qubit Hamiltonian expressed as a linear combination of Pauli operators:
\begin{equation}
\label{pauli_hamiltonian}
    H=\sum_{i=0}^{T-1}\alpha_iP_i,
\end{equation}
where $\alpha_i$ are real coefficients, $P_i$ are Pauli operators, and the number of terms is ${T=O(\text{poly}(n))}$.
We assume that the coefficients $\alpha_i$ are nonnegative, and instead each Pauli $P_i$ carries a $\pm1$ sign.
We also require a normalization condition on $H$: ${\sum_{i=0}^{T-1}|\alpha_i|=\sum_{i=0}^{T-1}\alpha_i=1}$.
Given some non-normalized input Hamiltonian whose coefficients are $\alpha_i^{(E)}$ (units of energy), we obtain normalized coefficients as ${\alpha_i=|\alpha_i^{(E)}|/\sum_{i=0}^{T-1}|\alpha_i^{(E)}|}$.

The block encoding unitary is
\begin{equation}
\label{unitary}
    U=\sum_{i=0}^{T-1}|i\rangle_a\langle i|_a\otimes P_i,
\end{equation}
i.e., application of each Pauli term $P_i$ controlled on the state of the auxiliary register being $|i\rangle_a$.
The corresponding block encoding state is
\begin{equation}
\label{G_a}
    |G\rangle_a=\sum_{i=0}^{T-1}\sqrt{\alpha_i}|i\rangle_a.
\end{equation}
Inserting $U$ as defined in \eqref{unitary} and $|G\rangle_a$ as defined in \eqref{G_a} into the left-hand side of \eqref{block_encoding_def} and using the fact that $\alpha_i\ge0$ for all $i$, one can verify that \eqref{block_encoding_def} is satisfied, i.e., this $(U,G)$ forms a valid block encoding of $H$.

One option for encoding $|i\rangle_a$ is to use $\lceil\log_2 T\rceil$ auxiliary qubits and let each $|i\rangle_a$ be the computational basis state corresponding to the binary number $i$~\cite{low2019qubitization}.
In this case, to implement $U$ we must, for each of the $T$ $P_i$'s, apply $P_i^{(j)}$ (the $j$th single-qubit Pauli operator in $P_i$) to system qubit $j$, controlled on the auxiliary qubits being in state $|i\rangle_a$.
$P_i$ is an $n$-qubit Pauli operator, so implementing $U$ requires applying at most $nT$ single-qubit Pauli operators, each controlled on all of the auxiliary qubits.

To prepare $|G\rangle_a=G|0\rangle_a$, we can use existing state preparation procedures, since there are only logarithmically-many auxiliary qubits and thus preparing an arbitrary state on them is efficient.
For example, one can use the state preparation based on binary data structures of \cite[Theorem A.1]{kerenidis2017recommendationsystems}.
To prepare an arbitrary real-amplitude state of $\lceil\log_2T\rceil$ qubits using this method requires fewer than $2T$ single-qubit rotations, each of which is controlled on up to all of the remaining qubits.

The qubitization iteration step that we will discuss below in \cref{qubitization_sec} also requires implementing $R$, a reflection around $|G\rangle_a$.
One can do this by applying $G^\dagger$, the inverse of the state preparation unitary, then applying the $\lceil\log_2T\rceil$-controlled phase that reflects around $|0\rangle_a$, and finally reapplying $G$.
Hence the total cost for $R$ is at most $4T$ $\lceil\log_2T\rceil$-controlled gates.
Methods for implementing the block encoding itself are not the main focus of this paper, but in \cref{block_encoding_app} we present an alternative encoding based on the binary representation of the Pauli operators, which uses more qubits in exchange for shorter circuits and is well-adapted to local Pauli Hamiltonians such as spin models.
When this block encoding is applied to a Heisenberg model containing arbitrary $XX$ and $ZZ$ interactions, $U$ requires $3n+T$ two-qubit gates, $G$ requires $4n^2-10n$ two-qubit gates plus $2$ single-qubit gates, and $R$ requires $8n^2+14$ two-qubit gates plus $4$ single-qubit gates.
Various other block encodings have been proposed that may be advantageous depending on the Hamiltonian and the available quantum device~\cite{low2019qubitization,steudtner2020blockencoding,camps2022blockencodings,camps2022fable}.

\subsection{Qubitization}
\label{qubitization_sec}

The other main component we will need is the simplest version of qubitization introduced in~\cite{low2019qubitization}, which applies when the block encoding unitary is self-inverse, i.e.,
\begin{equation}
\label{self_inv}
    U^2=\mathds{1}.
\end{equation}
In the example above, $U$ is the product of controlled Pauli operators given by \eqref{unitary}, so \eqref{self_inv} is satisfied.
The following lemma is based directly on Section IV in~\cite{low2019qubitization}; for completeness, we give its proof in \cref{proofs}. 

\begin{lemma}[Chebyshev polynomials from block encoding]
\label{chebyshev_lemma}
    Given a block encoding $(U,G)$ of a Hamiltonian $H$, such that $U^2=\mathds{1}$, let
    \begin{equation}
    \label{R_def}
        R\coloneqq (2|G\rangle_a\langle G|_a-\mathds{1}_a)\otimes\mathds{1}_s
    \end{equation}
    be the reflection around $|G\rangle_a$ in the auxiliary space.
    Then
    \begin{equation}
    \label{block_encoding_chebyshev}
        (\langle G|_a\otimes\mathds{1}_s)(RU)^k(|G\rangle_a\otimes\mathds{1}_s)=T_k(H)
    \end{equation}
    for any $k=0,1,2,...$, where $T_k(\cdot)$ is the $k$th Chebyshev polynomial of the first kind.
    In other words, $(RU)^k$ is a block encoding of $T_k(H)$.
\end{lemma}

\subsection{Quantum subspace diagonalization}

Quantum subspace diagonalization~\cite{mcclean2017subspace,colless2018computation,parrish2019filterdiagonalization,motta2020qite_qlanczos,takeshita2020subspace,huggins2020nonorthogonal,stair2020krylov,cohn2021filterdiagonalization,yoshioka2021virtualsubspace,epperly2021subspacediagonalization,seki2021powermethod,cortes2022krylov,klymko2022realtime,baek2022nonorthogonal,tkachenko2022davidson} is a method for obtaining a variational estimate of the ground state energy of a Hamiltonian $H$.
In the most general setting, we assume access to a set of states
\begin{equation}
\label{krylov_states}
    \{|\psi_k\rangle=U_k|\psi_0\rangle~|~k=0,1,2,...,D-1\}
\end{equation}
that can be prepared on a quantum computer by quantum circuits $U_k$.
First, we form the $D\times D$ matrices $\textbf{H}$ and $\textbf{S}$ whose entries are
\begin{equation}
\begin{split}
    &\textbf{H}_{ij}=\langle\psi_i|H|\psi_j\rangle=\langle\psi_0|U_i^\dagger HU_j|\psi_0\rangle,\\
    &\textbf{S}_{ij}=\langle\psi_i|\psi_j\rangle=\langle\psi_0|U_i^\dagger U_j|\psi_0\rangle.
\end{split}
\end{equation}
$\textbf{S}$ is the \emph{overlap} (Gram) matrix of the Krylov states \eqref{krylov_states}.
Both can be estimated by repeated swap~\cite{buhrman2001fingerprinting} or Hadamard tests~\cite{aharonov2006jonespolynomial}.
Then, having estimated $\textbf{H}$ and $\textbf{S}$, we classically solve the generalized eigenvalue problem
\begin{equation}
\label{gep}
    \textbf{H}\vec{v}=\epsilon\textbf{S}\vec{v}
\end{equation}
and find the lowest eigenvalue $\epsilon$.
This is our variational estimate of the ground state energy, and it corresponds to the lowest expected energy of any state in the subspace
\begin{equation}
    \text{span}\{|\psi_k\rangle=U_k|\psi_0\rangle~|~k=0,1,2,...,D-1\}.
\end{equation}
Note that typically the generalized eigenvalue problem \eqref{gep} must be regularized due to ill-conditioning of $\textbf{S}$ as $D$ increases~\cite{klymko2022realtime,epperly2021subspacediagonalization}: this is discussed in detail in \cref{error_analysis}.

The Lanczos method~\cite{lanczos1950iteration,cullum2002lanczos} is a classical matrix diagonalization method that is closely related to quantum subspace diagonalization.
In it, Krylov states $|\psi_i\rangle$ are obtained by multiplying some reference state $|\psi_0\rangle$ by powers of the Hamiltonian,
\begin{equation}
\label{lanczos}
    \{|\psi_k\rangle=H^k|\psi_0\rangle~|~k=0,1,2,...,D-1\},
\end{equation}
and these are used to construct the matrices $\textbf{H}$ and $\textbf{S}$.
This yields a ground state energy estimate that converges exponentially in $D$~\cite{saad1980lanczos}, assuming infinite precision arithmetic; the impact of noisy $\textbf{H}$ and $\textbf{S}$ is discussed in \cref{error_analysis}.
In general, however, applying this method to quantum systems is exponentially costly in the number of qubits, since the Krylov states have exponential dimension.

To avoid the classical barrier of exponential dimensionality, as noted above, several methods have been proposed to permit a quantum variant of the algorithm.
These include \emph{quantum Lanczos} (QLanczos)~\cite{motta2020qite_qlanczos}, in which the powers of $H$ are replaced by imaginary time evolutions under $H$, methods in which the powers of $H$ are replaced by real time evolutions under $H$~\cite{parrish2019filterdiagonalization,stair2020krylov,cohn2021filterdiagonalization,epperly2021subspacediagonalization,cortes2022krylov,klymko2022realtime}, and a method for approximating powers of $H$ using linear combinations of time evolutions~\cite{seki2021powermethod}.
All of these methods converge to the Lanczos method in appropriate limits.
However, either real or imaginary time evolutions are required, both of which entail approximations.

\section{Lanczos method on a quantum computer}
\label{method}

\subsection{Algorithm description}

We propose a quantum algorithm that exactly produces the Krylov space used in the classical Lanczos method up to finite sample noise.
We will focus on the case where the input block encoding is of a Hamiltonian encoded as a linear combination of Pauli operators as in the example in \cref{block_encoding}, since in this case the measurement scheme is particularly simple.
However, the method can be generalized to other block encodings, as we will describe.

The idea is to use as the Krylov basis vectors
\begin{equation}
\label{chebyshev_lanczos}
    \{|\psi_k\rangle=T_k(H)|\psi_0\rangle~|~k=0,1,2,...,D-1\}.
\end{equation}
The Chebyshev polynomials $T_k(x)$ for ${k=0,...,D-1}$ are a basis for the polynomials of degree less than $D$, so the Krylov space spanned by \eqref{chebyshev_lanczos} is identical to the Krylov space obtained from powers of the Hamiltonian:
\begin{equation}
\label{chebyshev_lanczos_subspace}
\begin{split}
    &\text{span}\{T_k(H)|\psi_0\rangle~|~k=0,1,2,...,D-1\}\\
    &=\text{span}\{H^k|\psi_0\rangle~|~k=0,1,2,...,D-1\}.
\end{split}
\end{equation}
Since the quantum subspace diagonalization method classically finds the state with lowest expected energy in this span, the performance of the Krylov space spanned by Chebyshev polynomials of the Hamiltonian applied to the initial state will be identical to that of powers of the Hamiltonian applied to the initial state, up to noise.

Also, we reviewed in \cref{qubitization_sec} (\cref{chebyshev_lemma}) how to implement block encodings of Chebyshev polynomials of a block encoded Hamiltonian.
There was no approximation in \cref{chebyshev_lemma}, so as long as the input block encoding is exact, the method will exactly produce the subspace \eqref{chebyshev_lanczos_subspace}.

It remains to show how to extract the matrices $\textbf{H}$ and $\textbf{S}$ from the block encodings of $T_k(H)|\psi_0\rangle$.
For $\textbf{S}$ we have
\begin{equation}
\label{overlap_1}
    \textbf{S}_{ij}=\langle\psi_0|T_i(H)T_j(H)|\psi_0\rangle=\big\langle T_i(H)T_j(H)\big\rangle_0,
\end{equation}
where $\langle\cdot\rangle_0$ denotes expectation value with respect to the initial state $|\psi_0\rangle$.
Chebyshev polynomials satisfy the identity
\begin{equation}
\label{chebyshev_product}
    T_i(x)T_j(x)=\frac{1}{2}\Big(T_{i+j}(x)+T_{|i-j|}(x)\Big).
\end{equation}
Inserting this into \eqref{overlap_1} yields
\begin{equation}
\label{overlap_2}
\begin{split}
    \textbf{S}_{ij}&=\frac{1}{2}\big\langle T_{i+j}(H)+T_{|i-j|}(H)\big\rangle_0\\
    &=\frac{1}{2}\Big(\big\langle T_{i+j}(H)\big\rangle_0+\big\langle T_{|i-j|}(H)\big\rangle_0\Big).
\end{split}
\end{equation}

For $\textbf{H}$ we have
\begin{equation}
\label{hamiltonian_1}
    \textbf{H}_{ij}=\big\langle T_i(H)HT_j(H)\big\rangle_0.
\end{equation}
Using the fact that $H=T_1(H)$ and applying \eqref{chebyshev_product} twice, we obtain
\begin{equation}
\label{hamiltonian_2}
\begin{split}
    \textbf{H}_{ij}=\frac{1}{4}\Big(&\big\langle T_{i+j+1}(H)\big\rangle_0+\big\langle T_{|i+j-1|}(H)\big\rangle_0\\
    &+\big\langle T_{|i-j+1|}(H)\big\rangle_0+\big\langle T_{|i-j-1|}(H)\big\rangle_0\Big).
\end{split}
\end{equation}
Since $i,j=0,1,2,...,D-1$, all matrix elements in both matrices are linear combinations of the expectation values
\begin{equation}
\label{necessary_exp_vals}
    \big\langle T_k(H)\big\rangle_0,\quad k=0,1,2,...,2D-1,
\end{equation}
where the highest value $2D-1$ comes from the first term in \eqref{hamiltonian_2} when $i=j=D-1$.
Therefore, in order to construct the matrices $\textbf{H}$ and $\textbf{S}$ it is enough to estimate all of the expectation values \eqref{necessary_exp_vals}.

Given our block encoding $(U,G)$ of $H$, \cref{chebyshev_lemma} yielded \eqref{block_encoding_chebyshev}, which when we take the expectation value with respect to $|\psi_0\rangle$ in the system space becomes
\begin{equation}
\label{chebyshev_exp_val_1}
    \big\langle T_k(H)\big\rangle_0=(\langle G|_a\otimes\langle\psi_0|)(RU)^k(|G\rangle_a\otimes|\psi_0\rangle).
\end{equation}
Since $R$ is a reflection about $|G\rangle_a$, it is Hermitian and also acts as identity on $|G\rangle_a$, so the leftmost $R$ in \eqref{chebyshev_exp_val_1} can be removed, yielding
\begin{equation}
\label{chebyshev_exp_val_2}
\begin{split}
    &\big\langle T_k(H)\big\rangle_0=(\langle G|_a\otimes\langle\psi_0|)U(RU)^{k-1}(|G\rangle_a\otimes|\psi_0\rangle).
\end{split}
\end{equation}
The product of operators in \eqref{chebyshev_exp_val_2} can be rewritten as
\begin{equation}
    U(RU)^{k-1}
    =
    \begin{cases}
        (UR)^{k/2}R(RU)^{k/2}\quad\text{if $k$ is even},\\
        (UR)^{\lfloor k/2\rfloor}U(RU)^{\lfloor k/2\rfloor}\quad\text{if $k$ is odd}.
    \end{cases}
\end{equation}
Therefore, if we define the state
\begin{equation}
\label{big_psi_def}
    |\psi_{\lfloor k/2\rfloor}\rangle=(RU)^{\lfloor k/2\rfloor}(|G\rangle_a\otimes|\psi_0\rangle),
\end{equation}
whose adjoint is
\begin{equation}
    \langle\psi_{\lfloor k/2\rfloor}|=(\langle G|_a\otimes\langle\psi_0|)(UR)^{\lfloor k/2\rfloor},
\end{equation}
then because $U$ and $R$ are both Hermitian, we can rewrite \eqref{chebyshev_exp_val_2} as
\begin{equation}
\label{chebyshev_exp_val_3}
\begin{split}
    &\big\langle T_k(H)\big\rangle_0=
    \begin{cases}
        \langle\psi_{\lfloor k/2\rfloor}|R|\psi_{\lfloor k/2\rfloor}\rangle\quad\text{if $k$ is even},\\
        \langle\psi_{\lfloor k/2\rfloor}|U|\psi_{\lfloor k/2\rfloor}\rangle\quad\text{if $k$ is odd}.
    \end{cases}
\end{split}
\end{equation}

Hence, we can estimate all the matrix elements of $\textbf{H}$ and $\textbf{S}$ by estimating the expectation values in \eqref{chebyshev_exp_val_3} for each $k=0,1,2,...,2D-1$.
We accomplish this as follows:
\begin{enumerate}
    \item Prepare $|\psi_{\lfloor k/2\rfloor}\rangle$ by applying $RU$ $\lfloor k/2\rfloor$ times to $|G\rangle_a\otimes|\psi_0\rangle$.
    
    \item If $k$ is even, we want to measure $R$. To do this, first apply $G^\dagger$ (recall that $|G\rangle_a=G|0\rangle_a$; see \cref{block_encoding}). Then measure $2|0\rangle_a\langle0|_a-\mathds{1}$ on the auxiliary qubits, i.e., measure all auxiliary qubits in the computational basis and return $+1$ if all outcomes are $|0\rangle$, otherwise $-1$.
    
    \item If $k$ is odd, we want to estimate the expectation value of
    \begin{equation}
        U=\sum_{i=0}^{T-1}|i\rangle_a\langle i|_a\otimes P_i.
    \end{equation}
    We do this by separately estimating the expectation value of each term, so for each state preparation we measure some $|i\rangle_a\langle i|_a\otimes P_i$.
    Measure the auxiliary qubits in the computational basis and the system qubits in a local Pauli basis compatible with $P_i$.
    If the outcome of this measurement is $|i\rangle_a$, return the system qubits' measurement outcome for $P_i$, and otherwise return $0$.
    
    \item Return to step 1 and repeat until enough measurements are obtained to estimate the desired expectation value to the desired precision. 
    
\end{enumerate}
\noindent
Note that in step 3 above, locally-compatible Pauli measurements can be grouped using any of the same methods as for standard Pauli Hamiltonian expectation value estimation~\cite{hadfield2022lbcs,huang2021derandomization,wu2023overlappedgrouping,hillmich2021decisiondiagrams,hadfield2021aps,shlosberg2023adaptiveestimation,yen2022deterministicmeasurement}, by measuring the system qubits in the shared local Pauli basis.
Also note that by employing one of these methods, the number of shots depends on the $l1$-norm of the Pauli coefficients in the Hamiltonian, which in our case is fixed to $1$ (see \cref{block_encoding}), so the requirement of measuring in distinct Pauli bases does not contribute additional measurement overhead asymptotically.
The above algorithm is displayed in \cref{alg_summary}.

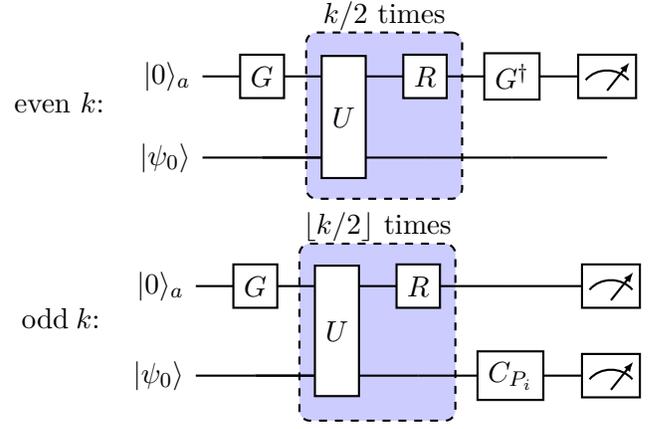
\begin{figure}
    even $k$:
    \begin{quantikz}
        \lstick{$|0\rangle_a$} & \gate{G} & \gate[2]{U}\gategroup[2,steps=2,style={dashed,
rounded corners,fill=blue!20, inner xsep=2pt},
background]{$k/2$ times} & \gate{R} & \gate{G^\dagger} & \meter{} \\
        \lstick{$|\psi_0\rangle$} & \qw & \qw & \qw & \qw & \qw
    \end{quantikz}\\\vspace{0.15in}
    odd $k$:
    \begin{quantikz}
        \lstick{$|0\rangle_a$} & \gate{G} & \gate[2]{U}\gategroup[2,steps=2,style={dashed,rounded corners,fill=blue!20, inner xsep=2pt},background]{$\lfloor k/2\rfloor$ times} & \gate{R} & \qw & \meter{} \\
        \lstick{$|\psi_0\rangle$} & \qw & \qw & \qw & \gate{\text{$C_{P_i}$}} & \meter{}
    \end{quantikz}
    \caption{The quantum circuits used in our algorithm. For each $k=0,1,2,...,2D-1$ we repeatedly sample from one of the above circuits, where $D$ is the desired Krylov space dimension. When $k$ is even, we apply the upper circuit. When $k$ is odd, we apply the lower circuit for each Pauli term $P_i$ in the Hamiltonian. All measurements are in the computational basis, and $C_{P_i}$ is a single layer of single-qubit Clifford gates that maps $P_i$ to a diagonal Pauli operator. As defined in the text, $U$ block encodes the Hamiltonian $H$, $G$ prepares the state that identifies the block containing $H$, and $R$ reflects about that state.}
    \label{alg_summary}
\end{figure}

\subsection{Runtime and qubit cost}
\label{runtime_analysis}

Let us summarize the requirements and costs for our quantum implementation of the Lanczos method:
\begin{enumerate}
    \item The required input is a block encoding $(U,G)$ of $H$: there are various options for constructing this, as discussed in \cref{block_encoding_sec}.
    Using the example in \cref{block_encoding_sec}, for $n$ qubits and $T$ Hamiltonian terms the cost of $U$ is $nT$, the cost of $G$ is $2T$, and the cost of $R$ is $4T$, all in $\lceil\log_2 T\rceil$-controlled single-qubit gates.
    
    \item The number of measurements depends on the desired error and the classical method for regularizing and solving the generalized eigenvalue problem \eqref{gep}. This analysis is provided in \cref{error_analysis}.

    \item Each state preparation [given in \eqref{big_psi_def}] requires preparing $|G\rangle_a\otimes|\psi_0\rangle$, followed by at most ${D-1}$ applications of $RU$. We then either measure in a local Pauli basis, or apply $G^\dagger$ and then measure. The latter leads to the longest coherent sequence of operations required by the algorithm. Using the costs of the block encoding in point 1 above, we find that this sequence requires at most ${(D-1)nT+4DT}$ $\lceil\log_2 T\rceil$-controlled single-qubit gates.
    
    \item The required qubits are those encoding the system, i.e., those that $H$ acts upon, plus auxiliary qubits whose state $|G\rangle_a$ identifies the block in which $H$ is encoded. The number of auxiliary qubits required depends on the choice of block encoding, but can be chosen to be $\lceil\log_2 T\rceil$ when the Hamiltonian contains $T$ terms, as discussed in \cref{block_encoding_sec}.

\end{enumerate}

\subsection{Other types of block encoding}

The algorithm above can be applied for any block encoding of a linear combination of Pauli operators as in \eqref{unitary}.
The method is independent of how the indices of the Pauli operators are encoded in the auxiliary register, but the specifics of the measurement scheme do depend on the terms being Pauli operators.
To see this, recall that when $k$ is odd, we want to estimate the expectation value of the block encoding unitary $U$ in the prepared state $|\psi_{\lfloor k/2\rfloor}\rangle$.
When the terms are Pauli operators, we accomplish this by separately estimating each term using the lower circuit in \cref{alg_summary}.

When a different block encoding is the input, however, we cannot necessarily measure $U$ by repeatedly measuring in Pauli bases.
In such a case, since by assumption $U$ is a unitary that we have a quantum circuit for, we can estimate its expectation value using a Hadamard test or equivalent (see for example~\cite{cortes2022krylov}).
This would replace the Pauli basis measurement in the odd $k$ circuit in \cref{alg_summary}.
Note that this is also an option in the case of a Pauli Hamiltonian, where it represents trading off increased circuit depth (for the Hadamard test or equivalent) for a lower number of repetitions (not having to repeat for all Pauli bases).

\section{Error analysis}
\label{error_analysis}

In this section we analyse the error scaling in our algorithm subject to finite sample noise and regularization of the overlap matrix involved in the generalized eigenvalue problem \eqref{gep}.
Finite sample noise enters because the matrix elements of $\textbf{H}$ and $\textbf{S}$ are obtained from expectation values estimated by repeated measurements.
This introduces a statistical error to the resulting energy estimates.

In addition, when executed on a real quantum computer any quantum algorithm will be subject to device noise.
Since our error analysis applies to arbitrary perturbations of the measured quantities, it can also be used to obtain energy error bounds under such device noise, but since that is device specific, we focus on finite sample noise for simplicity.

The overlap matrix $\textbf{S}$ must be regularized because the stability of the generalized eigenvalue problem \eqref{gep} depends on $\textbf{S}$ being well-conditioned.
In practice the condition number of $\textbf{S}$ grows exponentially with the Krylov space dimension $D$, which has been proven for a Krylov space generated by powers of the Hamiltonian~\cite{beckerman2017singularvalues}, and appears to hold numerically for Chebyshev polynomials of the Hamiltonian.
This means that the generalized eigenvalue problem must be regularized, which may be done by thresholding the eigenvalues of $\textbf{S}$, i.e., projecting both $\textbf{H}$ and $\textbf{S}$ onto the span of the eigenvectors of $\textbf{S}$ whose eigenvalues are above some threshold $\epsilon>0$~\cite{klymko2022realtime,epperly2021subspacediagonalization}.

For Krylov spaces generated by real time evolutions, the resulting error in the ground state energy estimate was analyzed in~\cite{epperly2021subspacediagonalization} (related techniques for handling noise and regularization are discussed in~\cite{klymko2022realtime}).
We adapt their proof to our Krylov space generated by Chebyshev polynomials of the Hamiltonian, with the following result:

\begin{theorem}
\label{error_theorem}
    Let $H$ be an $N\times N$ Hamiltonian with energies ${E_0\le E_1\le\cdots\le E_{N-1}}$ in the range $[-1,1]$.
    Suppose we compute the $D$-dimensional Krylov space spanned by \eqref{chebyshev_lanczos} using the method in \cref{method}, yielding noisy estimates
    \begin{equation}
        (\widetilde{\textbf{H}},\widetilde{\textbf{S}})=(\textbf{H}+\boldsymbol\Delta_H,\textbf{S}+\boldsymbol\Delta_S)
    \end{equation}
    of $(\textbf{H},\textbf{S})$, for some Hermitian perturbations $(\boldsymbol\Delta_H,\boldsymbol\Delta_S)$ whose spectral norms are $(\eta_H,\eta_S)$.
    Denote the total noise rate as
    \begin{equation}
        \eta\coloneqq\sqrt{\eta_H^2+\eta_S^2}.
    \end{equation}
    Let
    \begin{equation}
        \gamma_0=\langle E_0|\psi_0\rangle
    \end{equation}
    be the overlap of the initial reference state $|\psi_0\rangle$ with the true ground state $|E_0\rangle$.
    Let $\widetilde{\textbf{S}}$ be regularized by projecting it (and $\widetilde{\textbf{H}}$) onto the subspace spanned by its eigenvectors with eigenvalues above some threshold $\epsilon>0$.
    Let $\epsilon_\text{total}$ be the sum of the eigenvalues of $\textbf{S}$ discarded by regularizing it according to $\epsilon$ in the same way.
    Then provided the noise rate $\eta$ is sufficiently small (small enough that the assumptions of \cref{noise_theorem} can be satisfied), there exists a choice of threshold
    \begin{equation}
    \label{threshold_condition}
        \epsilon=\Omega\left((D^4\eta)^\frac{1}{1+\alpha}\right)
    \end{equation}
    for some constant $0\le\alpha\le1/2$ such that for any $\delta>0$, the error $\mathcal{E}$ in the ground state energy estimate coming from the regularized version of the noisy problem $(\widetilde{\textbf{H}},\widetilde{\textbf{S}})$ is bounded by
    \begin{equation}
    \label{total_error_bound}
    \begin{split}
        \mathcal{E}\le O\Bigg(&(D^4\eta)^{\frac{1}{1+\alpha}}+\frac{\sqrt{\delta}\,\epsilon_\text{total}}{|\gamma_0|^2}\\
        &\quad+\delta+\frac{1}{|\gamma_0|^2}\left(1+\frac{\delta}{2}\right)^{-D}\Bigg).
    \end{split}
    \end{equation}
\end{theorem}

The proof is given in \cref{lanczos_convergence_app}, which also gives a more precise characterization of what ``provided the noise rate $\eta$ is sufficiently small'' means.
The noise rate is not required to be exponentially small in any of the problem parameters.
Although \eqref{total_error_bound} is given as an asymptotic scaling for simplicity, the two results it is based on (\cref{noise_theorem} and \cref{threshold_theorem} in \cref{lanczos_convergence_app}) give explicit bounds on the errors due to noise and thresholding, respectively, so the total error bound \eqref{total_error_bound} could also be reformulated as an explicit (i.e., nonasymptotic) bound.

The result \eqref{total_error_bound} may be interpreted as follows.
The last term is the error due to the exact (i.e., noiseless and nonthresholded) Krylov space, which vanishes exponentially quickly with the Krylov space dimension $D$.
The third term, $\delta$, is an energy error tolerance that determines the rate of convergence of the last term.
The error bound \eqref{total_error_bound} holds for any $\delta>0$, i.e., $\delta$ is only relevant to the analysis, and does not actually impact the algorithm.
It characterizes the rate of exponential decay of amplitudes of energies more than $\delta$ above the ground state energy, via the final term in \eqref{total_error_bound}.
Hence if $\delta$ is chosen to be equal to the spectral gap $\Delta$, then the third term in \eqref{total_error_bound} can be removed, because in that case there are no excited states within $\delta$ of the ground state energy.
However, if $\delta$ is chosen to be larger than the spectral gap, then although the ground state energy is approximated, the corresponding state will not approximate the true ground state in general: it will only approximate some arbitrary state in the low-energy subspace within $\delta$ of the ground state energy.

The second term in \eqref{total_error_bound} is due to regularization of $\widetilde{\textbf{S}}$ by $\epsilon$, i.e., to the discarding of eigenspaces of $\widetilde{\textbf{S}}$ with eigenvalues smaller than $\epsilon$.
Finally, the first term in \eqref{total_error_bound} is due to finite sample noise.
The factor of $D^{\frac{4}{1+\alpha}}$ in this term results from the proof technique of Theorem 2.7 in \cite{epperly2021subspacediagonalization} (reproduced in our notation as \cref{noise_theorem} in \cref{lanczos_convergence_app}), in which the authors note that this factor might be reduced or eliminated by an improved proof; we discuss this further below.
Numerics also indicate linear scaling $\epsilon=\Theta(\eta)$ (i.e., $\alpha=0$) if various heuristics are adopted to select an optimal threshold $\epsilon$ [note that this scaling is compatible with \eqref{threshold_condition}].
Such heuristics involve starting at a high threshold and decreasing it as long as the energy estimate continues to converge.

To understand the implications of \cref{error_theorem} in practice, we can begin by observing that since the eigenvalues of $\widetilde{\textbf{S}}$ decrease exponentially, as mentioned above,
\begin{equation}
\label{epsilon_relations}
    \epsilon_\text{total}=\Theta(\epsilon).
\end{equation}
Also, when the noise $\eta$ comes from finite numbers of samples to obtain each matrix element in $(\widetilde{\textbf{H}},\widetilde{\textbf{S}})$,
\begin{equation}
\label{meas_relation}
    \eta=\Theta(1/\sqrt{M}),
\end{equation}
when $M$ measurements are performed per matrix element.

Note that a complete expression for $\eta$ would also include some dependence on $D$, since $\eta$ is obtained from spectral norms of the perturbation matrices $(\boldsymbol\Delta_H$ and $\boldsymbol\Delta_S)$.
However, $\eta$ enters \eqref{total_error_bound} in the first two terms via \eqref{threshold_condition} and \eqref{epsilon_relations}, and in the context of this total energy error bound, no strong dependence on $D$ is generally found in practice, whether coming from $\eta$ or from the explicit factors in \eqref{threshold_condition}.
While in the presence of noise and thresholding, it sometimes happens that the error increases over short ranges of $D$, these increases are relatively small and the overall behavior is still bounded below a monotonically decreasing envelope.
See our numerics in \cref{numerics} and the numerics in \cite{klymko2022realtime,epperly2021subspacediagonalization}, for example.
Providing a rigorous proof of this property is an interesting direction for future work.

Finally, although \cite{epperly2021subspacediagonalization} provided numerics suggesting that $\alpha=1/4$ is a good value, that was for the original bound through which $\alpha$ entered their proof, and our numerics suggest that in practice, at least in the context of \cref{error_theorem} for our Chebyshev polynomial subspace, $\alpha=0$ is a better value [see \cref{numerics}].
Hence, using \eqref{threshold_condition} with $\alpha=0$, \eqref{epsilon_relations}, and \eqref{meas_relation}, we obtain
\begin{equation}
    \epsilon_\text{total}=\Theta(\epsilon)=\Theta(\eta)=\Theta\left(\frac{1}{\sqrt{M}}\right),
\end{equation}
which summarizes the ``in practice'' relations between these quantities.
Subject to the above assumptions, \eqref{total_error_bound} becomes the ``in practice'' error bound
\begin{equation}
\label{total_error_bound_practice}
    \mathcal{E}\le O\left(\frac{1}{\sqrt{M}}+\frac{\sqrt{\delta}}{|\gamma_0|^2\sqrt{M}}+\delta+\frac{1}{|\gamma_0|^2}\left(1+\frac{\delta}{2}\right)^{-D}\right).
\end{equation}

From \eqref{total_error_bound_practice}, we can find the Krylov space dimension $D$ required to reach error $\mathcal{E}$ in the ground state energy estimate.
From the final term in \eqref{total_error_bound_practice}, in the limit of small initial state overlap $|\gamma_0|$ and $\delta$ we obtain
\begin{equation}
\label{krylov_dimension_required_init}
    D=\Theta\left(\frac{\log\frac{1}{|\gamma_0|}+\log\frac{1}{\mathcal{E}}}{\delta}\right).
\end{equation}

To simplify this scaling, note that since $\delta$ appears as a term in the error \eqref{total_error_bound_practice} whenever it is larger than the spectral gap $\Delta$ (as discussed above), if the target error $\mathcal{E}$ is larger than $\Delta$, we must choose $\delta=O(\mathcal{E})$.
Given this, there is no reason to choose it to be smaller, so we can choose ${\delta=\Theta(\mathcal{E})}$.
Substituting this into \eqref{krylov_dimension_required_init} yields
\begin{equation}
\label{krylov_dimension_required}
    D=\Theta\left(\frac{\log\frac{1}{|\gamma_0|}+\log\frac{1}{\mathcal{E}}}{\mathcal{E}}\right).
\end{equation}

If instead the target error $\mathcal{E}$ is smaller than $\Delta$, then we should instead choose $\delta=\Delta$, since in that case the $\delta$ term in \eqref{total_error_bound_practice} (and \eqref{total_error_bound}) vanishes as noted above.
Hence in general
\begin{equation}
    \delta=\Theta(\max(\mathcal{E},\Delta)),
\end{equation}
so a general bound on the required Krylov space dimension $D$ is
\begin{equation}
\label{krylov_dimension_required_2}
    D=\Theta\left[\left(\log\frac{1}{|\gamma_0|}+\log\frac{1}{\mathcal{E}}\right)\min\left(\frac{1}{\mathcal{E}},\frac{1}{\Delta}\right)\right].
\end{equation}
Recall that $D$ is also the maximum circuit depth in terms of queries to the block encoding operators (see \eqref{chebyshev_exp_val_3} and \cref{alg_summary}).

The first two terms in \eqref{total_error_bound_practice} can then provide an asymptotic lower bound on the number of measurements $M$.
From the first term, we find that
\begin{equation}
    M=\Omega\left(\frac{1}{\mathcal{E}^2}\right),
\end{equation}
which should be unsurprising since this is always the scaling obtained from algorithms based on repeated sampling.
From the second term, we obtain a lower bound in terms of the initial state overlap $|\gamma_0|^2$:
\begin{equation}
    M=\Omega\left(\frac{\delta}{\mathcal{E}^2|\gamma_0|^4}\right)=\Omega\left(\frac{1}{\mathcal{E}|\gamma_0|^4}\right),
\end{equation}
where the second equality follows using ${\delta=\Theta(\mathcal{E})}$ as above.
Since $M$ must satisfy both of these asymptotic lower bounds, we can combine them as
\begin{equation}
    M=\Theta\left(\frac{1}{\mathcal{E}^2}+\frac{1}{\mathcal{E}|\gamma_0|^4}\right),
\end{equation}
the total number of measurements required per matrix element.
Finally, since all matrix elements of $\widetilde{\textbf{S}}$ and $\widetilde{\textbf{H}}$ are calculated from the $2D$ expectation values \eqref{necessary_exp_vals}, the total number of measurements required for the entire algorithm is
\begin{equation}
\label{measurements_lower_bound}
    DM=\widetilde{\Theta}\left(\left(\frac{1}{\mathcal{E}^2}+\frac{1}{\mathcal{E}|\gamma_0|^4}\right)\min\left(\frac{1}{\mathcal{E}},\frac{1}{\Delta}\right)\right),
\end{equation}
where the $\widetilde{\Theta}$ indicates that the logarithmic factor in \eqref{krylov_dimension_required_2} has been suppressed.

For a nonasymptotic analysis of this scaling, one would need to account for covariance of the resulting matrix element estimates, but this is unnecessary to obtain the asymptotic scaling \eqref{measurements_lower_bound}.
The first term in \eqref{total_error_bound_practice} might in principle also have some bearing on the number of measurements as a function of the Krylov space dimension, but as noted above, \cite{epperly2021subspacediagonalization} points out that this scaling may be an artifact of their proof technique.
Regardless, the number of measurements depends only polynomially on the Krylov space dimension as well as the problem parameters $|\gamma_0|$ and $\mathcal{E}$.

\section{Numerical demonstrations}
\label{numerics}

\begin{figure}[htp!]
    \centering
    \includegraphics[width=\columnwidth]{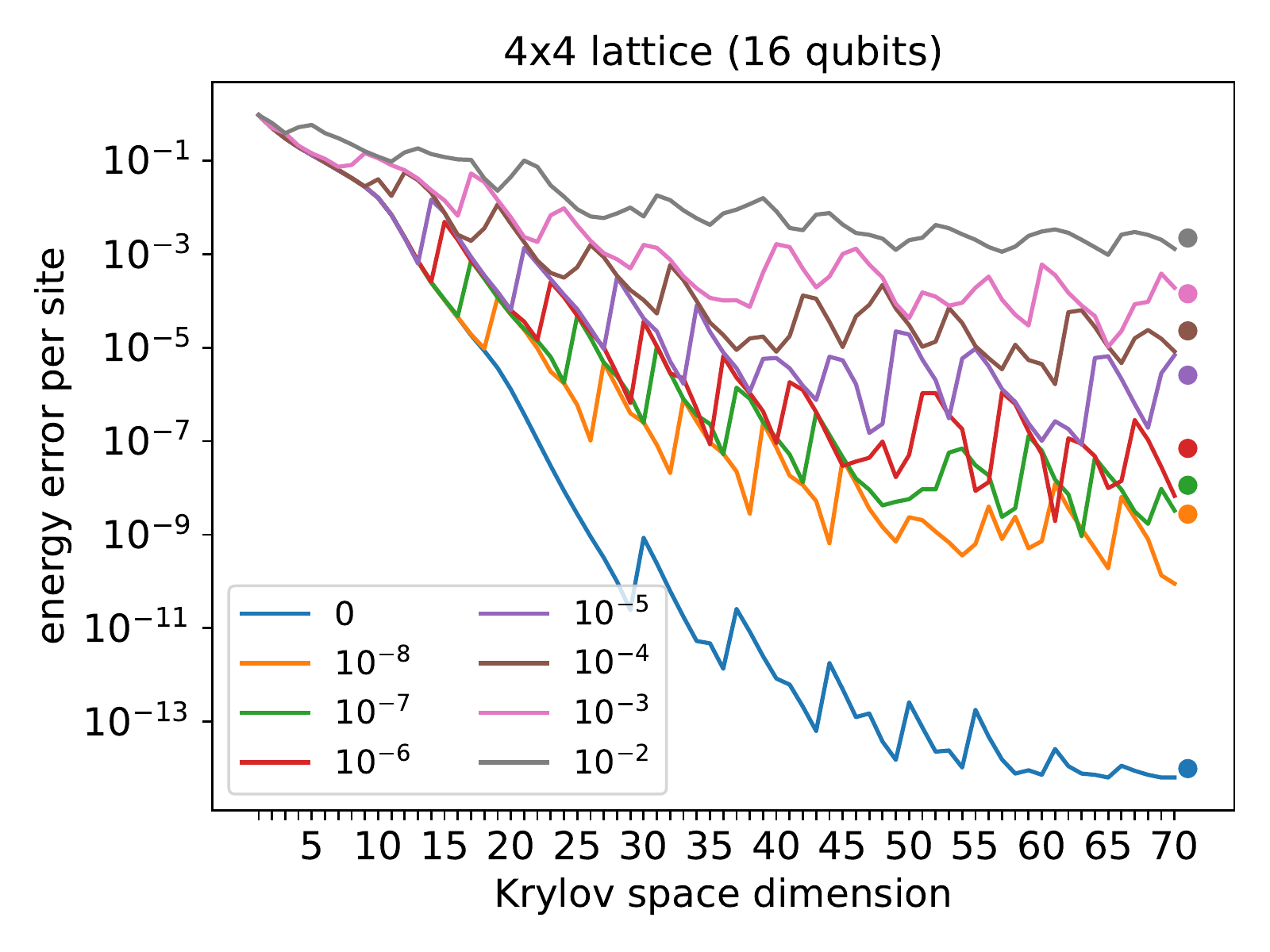}
    \includegraphics[width=\columnwidth]{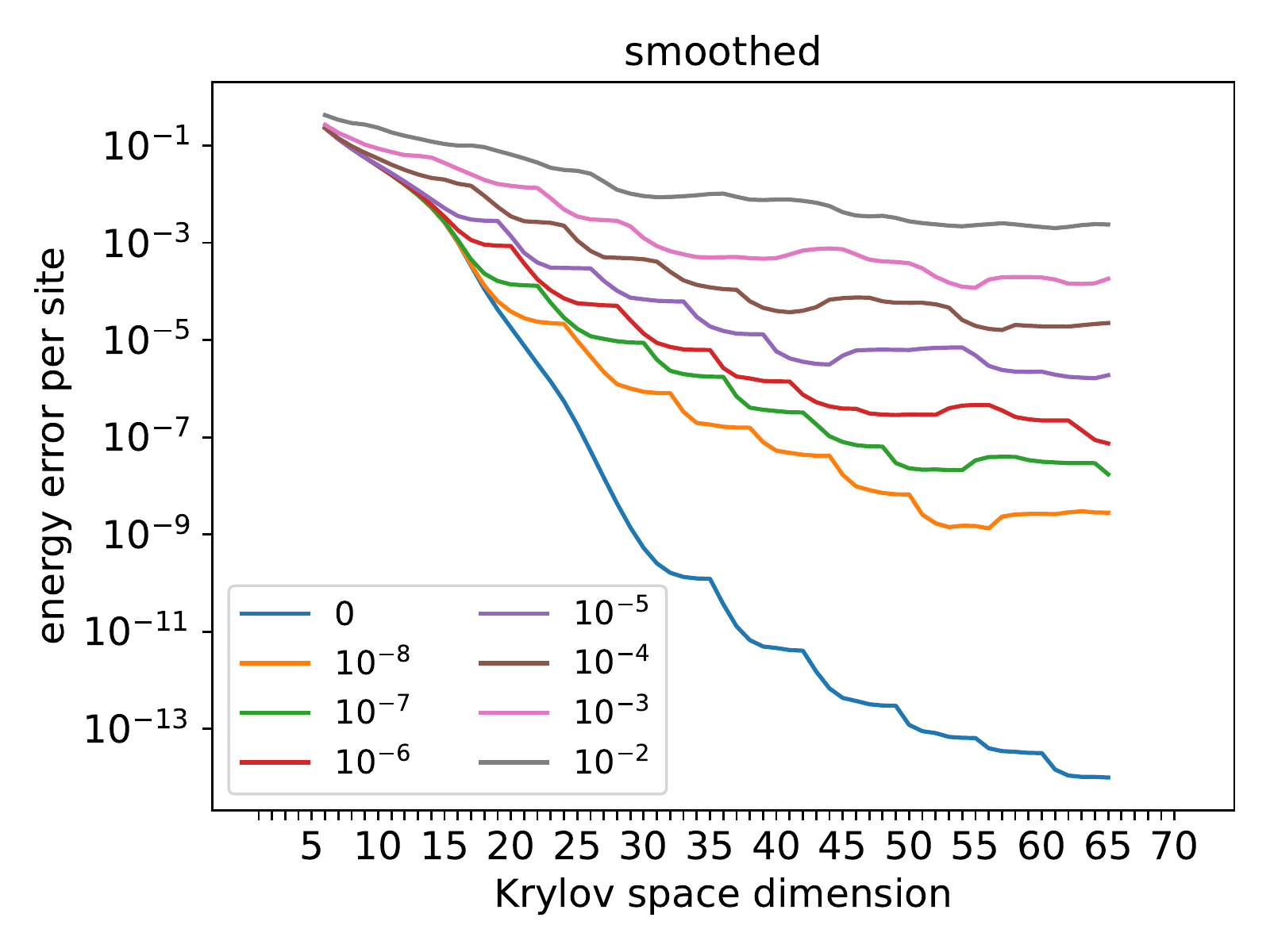}
    \caption{Results of numerical simulations of the nonorthogonalized Lanczos method implemented by our algorithm, for the J1-J2 model on a $4\times4$ square lattice with couplings $J_1=1$ and $J_2=0.5$. The upper panel shows energy error per site versus Krylov space dimension. The different curves correspond to different noise rates (given in the legend), which were imposed as Gaussian noise on the expectation values that would be estimated in the experiment [in \eqref{necessary_exp_vals}] to simulate the effect of finite sample noise. The points to the right of the curves are the means of the rightmost 10 points. Since the upper panel might be difficult to read, the lower panel represents the same data, but smoothed by taking the mean of each successive sequence of 10 points (each mean plotted above the center of its corresponding sequence).}
    \label{spin_error_vs_dim}
\end{figure}

As a demonstration of our method, we classically simulated its Krylov spaces for some example Hamiltonians.
For each Hamiltonian, we directly calculated the expectation values in \eqref{necessary_exp_vals}, which are the quantities we would actually estimate on the quantum computer in a real experiment, then used them to construct the subspace matrices $\textbf{H}$ and $\textbf{S}$ as in \eqref{overlap_2} and \eqref{hamiltonian_2}.
We simulated the effect of finite sample noise by adding Gaussian noise with various standard deviations to those expectation values.
We then regularized the generalized eigenvalue problem by thresholding as described in the previous section.

\begin{figure}[htp!]
    \centering
    \includegraphics[width=\columnwidth]{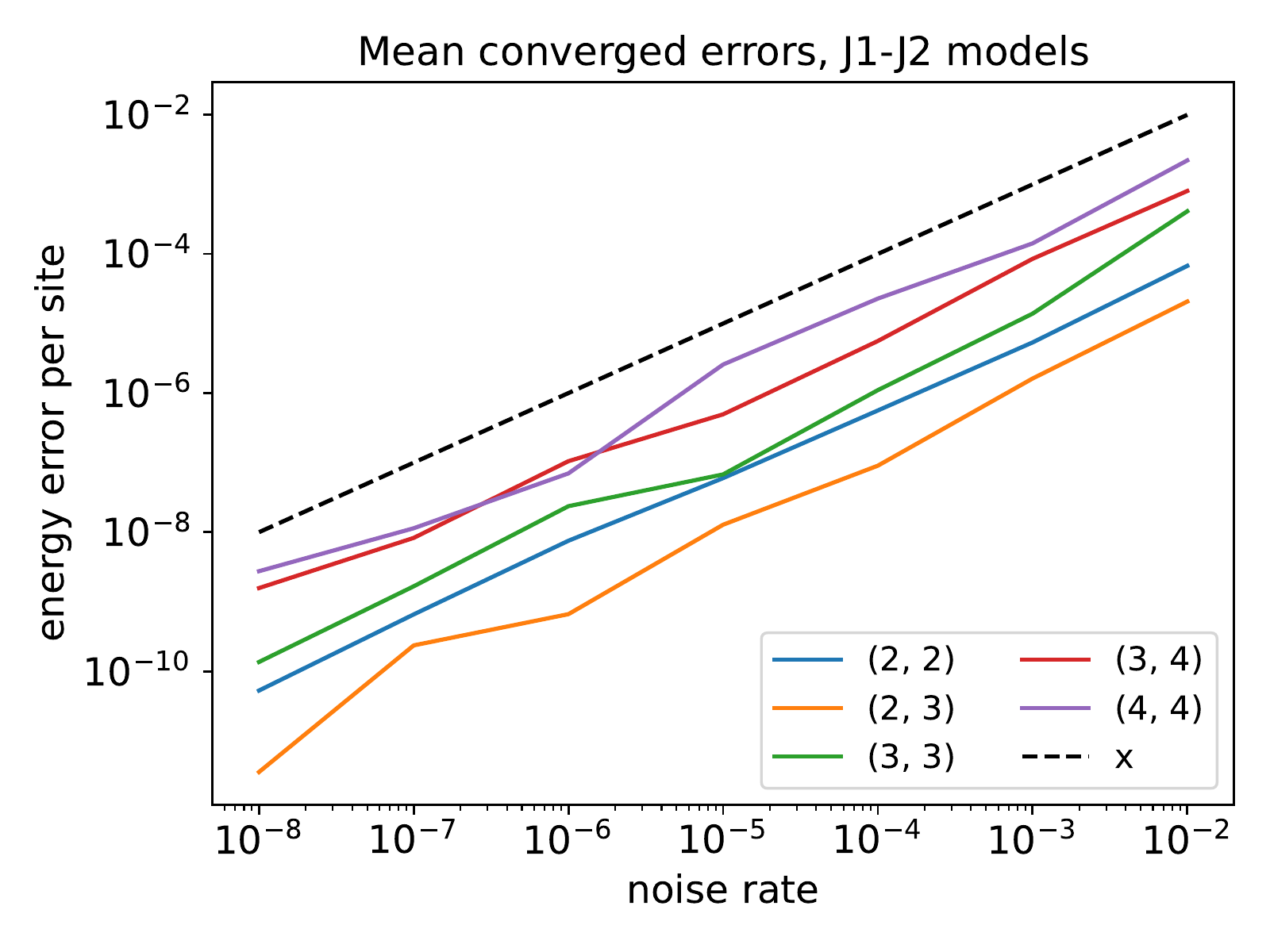}
    \includegraphics[width=\columnwidth]{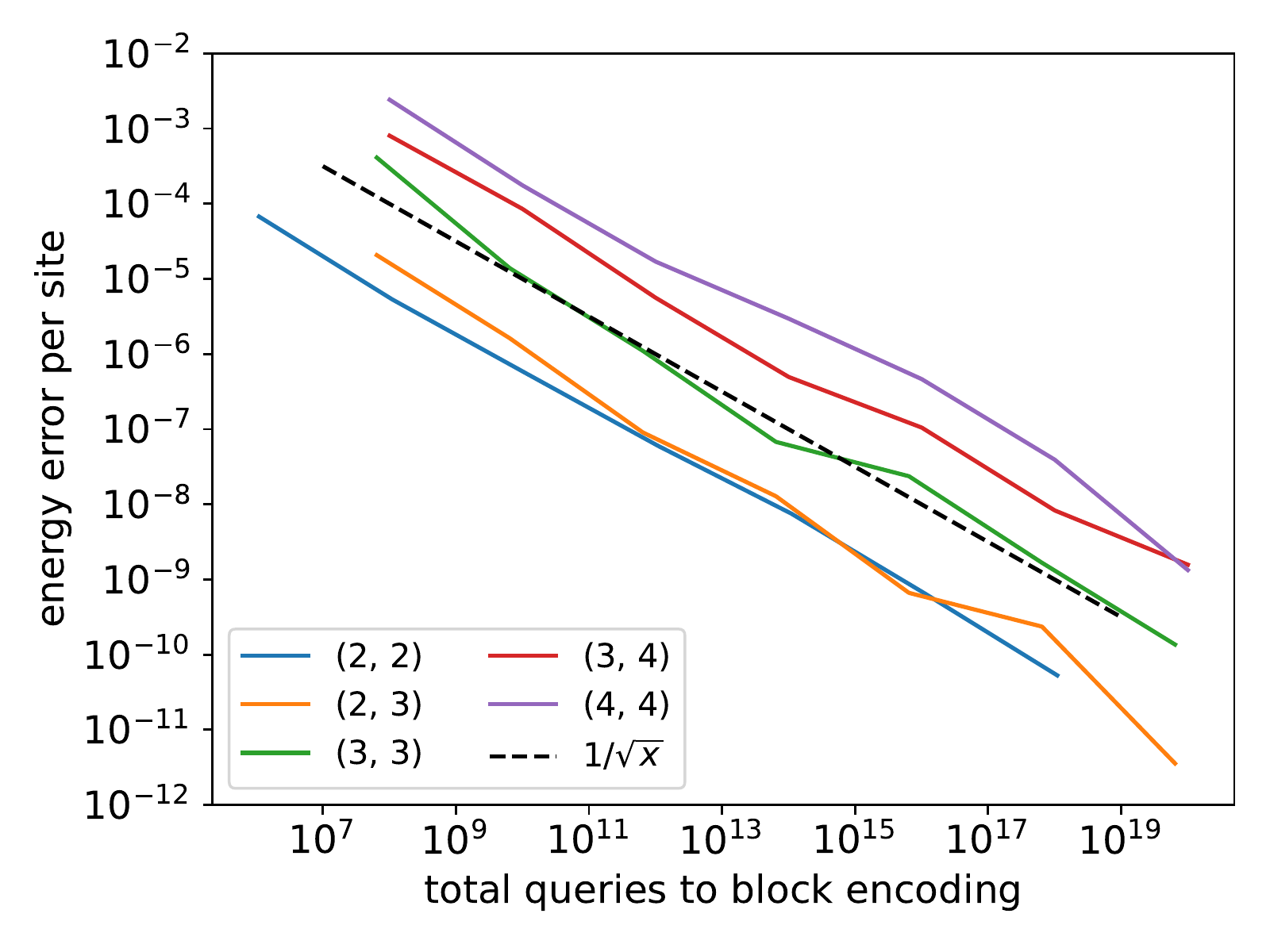}
    \caption{Energy error per site versus noise rate per measured quantity (upper panel) and total queries to block encoding (lower panel) for the J1-J2 model on several different square lattices given in the legend. These energy errors were each averaged over a sequence of 10 dimensions once convergence was reached (i.e., once the contribution of the last term in \eqref{total_error_bound} had become negligible); for example, the purple curve in the upper panel is the points on the right of the upper panel in \cref{spin_error_vs_dim}, plotted against the corresponding noise rates. The dashed lines illustrate the expected scalings from the theoretical analysis. Note that the two panels show the same information since the total queries to the block encoding is proportional to the total number of shots times the circuit depth, which were chosen to be the depths required to reach convergence. For both plots above we took those depths to be 5 for $(2,2)$, 40 for $(2,3)$ and $(3,3)$, and 50 for $(3,4)$ and $(4,4)$.}
    \label{error_vs_noise}
\end{figure}

For zero noise simulations, the threshold was set to $10^{-13}$.
For simulations with nonzero noise, the threshold was set proportional to the noise rate, with the constants being 30 for the spin models and 50 for molecules.
These constants were chosen to eliminate most spurious eigenvalues due to ill-conditioning: such spurious eigenvalues appear as dramatic downward jumps in the energy estimates with increasing subspace dimension.
In an actual experiment one would want to choose the thresholds by detecting such jumps automatically, but for the purpose of our simple demonstration it was enough to choose constants by hand that eliminated most of the jumps.
Since some spurious eigenvalues remained, we further eliminated outliers by taking as our datapoints the means of the middle 10\% of 100 independent runs at each noise rate and subspace dimension.

\begin{figure}[htp!]
    \centering
    \includegraphics[width=\columnwidth]{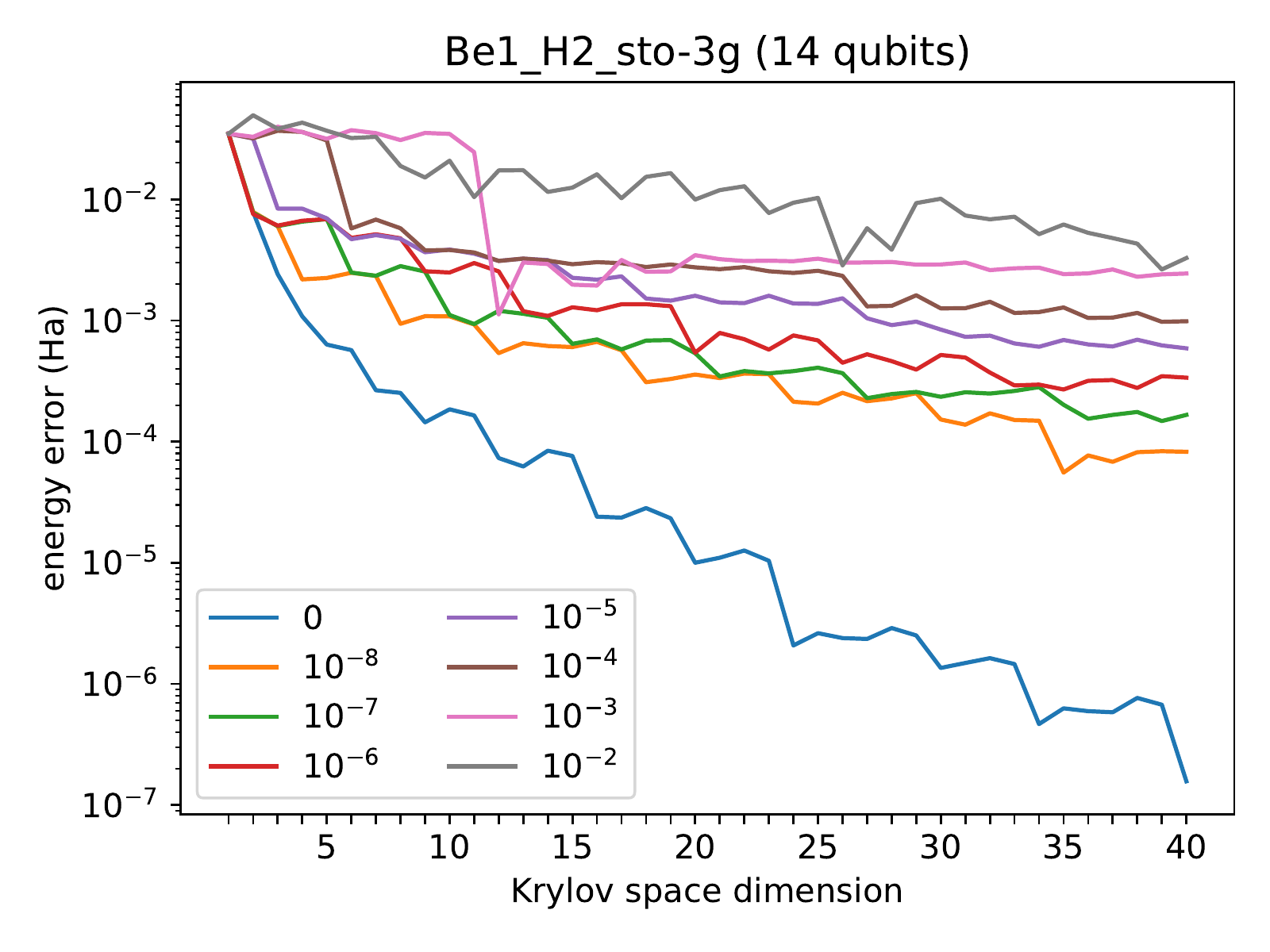}
    \caption{Results of numerical simulations of the nonorthogonalized Lanczos method implemented by our algorithm, for BeH$_2$ in the STO-3G basis, equilibrium configuration (bondlength $1.3264$ \AA). The plot shows energy error per site versus Krylov space dimension. As in \cref{spin_error_vs_dim}, the different curves correspond to different noise rates (given in the legend), which were imposed as Gaussian noise on the expectation values that would be estimated in the experiment [in \eqref{necessary_exp_vals}] to simulate the effect of finite sample noise. This is an all-electron calculation (no frozen core).}
    \label{es_error_vs_dim}
\end{figure}

Results of these simulations are shown in \cref{spin_error_vs_dim,error_vs_noise,es_error_vs_dim}.
\cref{spin_error_vs_dim} shows energy error as a function of Krylov space dimension for the J1-J2 spin model on a $4\times4$ square lattice.
The Hamiltonian of the J1-J2 model is given by
\begin{equation}
    H=J_1\sum_{\langle i,j\rangle}\vec{S}_i\cdot\vec{S}_j+J_2\sum_{\langle\langle i,j\rangle\rangle}\vec{S}_i\cdot\vec{S}_j,
\end{equation}
where $\langle\cdot,\cdot\rangle$ denotes nearest-neighbor pairs, $\langle\langle\cdot,\cdot\rangle\rangle$ denotes next-nearest-neighbor pairs, and the couplings in our case were set to $J_1=1$ and $J_2=0.5$.
The initial state was chosen to be the antiferromagnetic state, i.e., a checkerboard of alternating spins across the lattice, which yielded an initial state overlap of $|\gamma_0|=0.179$.
The upper panel shows the complete data, while for easier reading the lower panel shows a smoothed version of the same curves obtained by averaging each set of energy errors at 10 successive subspace dimensions.
In both cases, after a period of ill-defined behavior at low Krylov space dimension, the overall trends at all noise rates settle into roughly exponential decay towards some minimum ``converged'' error level, with local fluctuations along the way.
This is expected based on the theoretical error bounds in the previous section: the error decays exponentially with Krylov space dimension towards a minimum value set by the first three terms in \eqref{total_error_bound}.

\cref{error_vs_noise} shows energy error versus simulated noise rate for the J1-J2 model on a number of different square lattices.
These energy errors were each averaged over a sequence of 10 dimensions once convergence was reached (i.e., once the contribution of the last term in \eqref{total_error_bound} had become negligible).
As we would expect, at this point energy error and noise rate are related approximately linearly (the dashed line shows the curve $y=x$).

Finally, \cref{es_error_vs_dim} shows energy error versus Krylov space dimension for the electronic structure Hamiltonian BeH$_2$ in the STO-3G basis.
The initial state was chosen to be the Hartree-Fock state, which has an initial state overlap of $0.986$.
Unfortunately, numerical instability in our classical simulation prevented us from reaching as high Krylov space dimension for the electronic structure Hamiltonian as for the spin models, and it is not clear that the energies had finished converging to their minimum values.
This is possibly due to the choice of the Hartree-Fock state as the initial state: although it is a good initial approximation to the ground state in this case, the Hartree-Fock state contains no electron correlation.
This means that all electron correlation of the lowest-energy state in the Krylov space must be contributed by the corresponding polynomial of the Hamiltonian, which may live in a poorly conditioned part of the space.

Since we have access only to classically simulable small models, we cannot draw strong conclusions from these simulations regarding the asymptotic behavior of our method.
However, at least for these test cases, the behavior is roughly what we would expect from the theory in the previous section.

\section{Conclusion}
\label{conclusion}

In this work we have shown how to exactly reproduce the Krylov space of the classical Lanczos method on a quantum computer, up to sampling error, while avoiding the exponential classical cost of representing quantum states.
The required Krylov space dimension is inverse polynomial in the desired ground state energy error, and the sampling overhead is inverse polynomial in both the ground state energy error and the initial state overlap.
To our knowledge, the present method is conceptually the simplest use of block encoding in quantum algorithms to date, since it only requires applying up to one local basis rotation per circuit in addition to the block encoding unitaries themselves (see \cref{alg_summary}).

We conclude by discussing how the present method compares with existing quantum algorithms for approximating ground states.
As noted in the introduction, if we compare to previous quantum subspace diagonalization methods~\cite{motta2020qite_qlanczos,parrish2019filterdiagonalization,stair2020krylov,cohn2021filterdiagonalization,seki2021powermethod,cortes2022krylov,klymko2022realtime}, our algorithm benefits from being exactly equivalent to the classical Lanczos method up to sampling noise, without requiring approximate real or imaginary time evolution.
This offers two advantages.
First, the error in approximating real or imaginary time evolution, even if it can be suppressed, can never be reduced to zero.
Second, in quantum algorithms for simulating real time evolution, the circuit depths increase as the target error is reduced.
There are many methods to approximate time evolutions, but we mention two to illustrate this point.
First, if the time evolutions are approximated by Trotterization, then they require depth of order $\epsilon^{-1/k}$ where $k$ is the Trotter order and $\epsilon$ is the error.
Hence deep circuits will be required merely to reach high accuracy in approximating the Krylov states: if $d$ digits of accuracy are desired, then the resulting circuit depths scale as $O(\exp(d/r))$.

On the other hand, if the time evolutions are approximated by qubitization~\cite{low2019qubitization}, the resulting circuit depths scale much more favorably with error as $O(\log(1/\epsilon))$.
However, this case is easily compared to our algorithm, because qubitization uses the same block encoding unitaries that our method is based upon.
Thus, while qubitization would require circuits using up to $O(D+\log(1/\epsilon))$ applications of the block encoding unitaries, as well as other controlled phases, to construct a Krylov space of dimension $D$, our method requires circuits of depth at most $D$ in the same block encoding unitaries, with the only additional gates being the local basis rotation mentioned above.

This means that in practice, the advantage of our method over prior quantum subspace expansions is that it avoids all question of the accuracy of approximation of the Krylov states themselves, completely eliminating this source of error and circuit depth.
This reduces both depth and number of measurements required to obtain a result of a given accuracy.
Since circuit depth and number of measurements are key barriers to the success of quantum algorithms and will remain so even into the fault-tolerant era, this advance is significant.

In comparison to using quantum imaginary time evolution to approximate the ground state, we do not require the underlying quantum state to have finite correlation length as it goes through the evolution.
Even in the presence of a finite correlation length, there are significant prefactors in imaginary time evolution that lead to high computational cost (true, the cost is $4^{O(\log n)}=O(\text{poly}(n))$, but this can be a high-degree polynomial~\cite{motta2020qite_qlanczos}).
Adiabatic state preparation~\cite{farhi2000adiabatic} requires a gapped path from some efficiently-solvable Hamiltonian to the target Hamiltonian, while our method is independent of the spectral gap when that is small (see \cref{error_analysis}).
It also avoids the typical pitfalls of parametric optimization that appear in variational quantum eigensolvers~\cite{peruzzo2014vqe}, such as local minima and barren plateaus~\cite{mcclean2018barren,uvarov2021barren,cerezo2021barren,wang2021noiseinduced,liu2022laziness}, and admits rigorous error bounds (\cref{error_analysis}).

The most well-known algorithm for quantum computation of ground state energies is phase estimation.
To be successful, both our method and phase estimation require an initial state overlap $|\gamma_0|=\Omega(1/\text{poly}(n))$, for different reasons.
For phase estimation, the state overlap determines the success probability of the algorithm, while for our method it impacts the number of measurements required to obtain small errors (only entering the circuit depth logarithmically).
Similar to other quantum subspace diagonalization methods, our algorithm has quantifiable robustness to noise in the sense that it produces energy estimates with bounded error even if the subspace matrices are not exact.
This robustness is not encountered in conventional phase estimation (although more recent ground-state preparation algorithms such as~\cite{lin2022heisenberglimited} do possess some noise resilience).
On the other hand, phase estimation achieves Heisenberg scaling with error, i.e., $O(1/\mathcal{E})$, while our algorithm possesses the typical $O(1/\mathcal{E}^2)$ scaling of methods with finite sample noise.
Hence, although phase estimation may be preferable to reach extremely high precision energy estimates in the long run, many specific problems of interest have fixed error demands depending on their application.
For such problems, in the intermediate term while robustness to noise is still the crucial provision, our method presents a useful option.

~
\begin{acknowledgements}
The authors thank Ethan Epperly for sharing his expertise on analysis of Krylov methods, Cristian Cortes for correctly suggesting that a tighter error bound might be available than we gave in the first version of the manuscript, and Giuseppe Carleo, Guglielmo Mazzola, Andreas L\"auchli, Garnet Chan, Lin Lin, and Sergey Bravyi for providing helpful feedback on early drafts.
W.~M.~K. acknowledges support from the National Science Foundation, Grant No. DGE-1842474.
\end{acknowledgements}


\bibliographystyle{quantum}
\bibliography{references}

\appendix

\section{Proof of \cref{chebyshev_lemma}}
\label[appendix]{proofs}

The following lemma is based directly on Section IV in~\cite{low2019qubitization}:

\textbf{Lemma~\ref{chebyshev_lemma}} (Chebyshev polynomials from block encoding)\textbf{.}
\emph{
    Given a block encoding $(U,G)$ of a Hamiltonian $H$, such that $U^2=\mathds{1}$, let
    \begin{equation}
    \label{R_def_app}
        R\coloneqq (2|G\rangle_a\langle G|_a-\mathds{1}_a)\otimes\mathds{1}_s
    \end{equation}
    be the reflection around $|G\rangle_a$ in the auxiliary space.
    Then
    \begin{equation}
    \label{block_encoding_chebyshev_app}
        (\langle G|_a\otimes\mathds{1}_s)(RU)^k(|G\rangle_a\otimes\mathds{1}_s)=T_k(H)
    \end{equation}
    for any $k=0,1,2,...$, where $T_k(\cdot)$ is the $k$th Chebyshev polynomial of the first kind.
    In other words, $(RU)^k$ is a block encoding of $T_k(H)$.
}
\begin{proof}

Let $|\lambda\rangle$ denote an arbitrary eigenvector of $H$ with eigenvalue $\lambda$.
We can evaluate the action of $U$ on $|G\rangle_a\otimes|\lambda\rangle$ as follows:
\begin{equation}
\label{lemma1step1}
    U(|G\rangle_a\otimes|\lambda\rangle)=|G\rangle_a\otimes H|\lambda\rangle+\eta_\lambda|\perp_\lambda\rangle,
\end{equation}
where $\eta_\lambda$ is some positive number that preserves normalization and $|\perp_\lambda\rangle$ is some (normalized) state of all of the qubits that is orthogonal to $|G\rangle_a$ in the auxiliary space, i.e., such that
\begin{equation}
\label{perp_condition}
    (\langle G|_a\otimes\mathds{1}_s)|\perp_\lambda\rangle=0.
\end{equation}
Hence, \eqref{lemma1step1} simply represents resolving the state $U(|G\rangle_a\otimes|\lambda\rangle)$ into its components parallel and orthogonal to $|G\rangle_a$ in the auxiliary space: the latter ($|\perp_\lambda\rangle$) is unknown, but the former is
\begin{equation}
\begin{split}
    &(|G\rangle_a\langle G|_a\otimes\mathds{1}_s)U(|G\rangle_a\otimes|\lambda\rangle)\\
    &=(|G\rangle_a\otimes\mathds{1}_s)(\langle G|_a\otimes\mathds{1}_s)U(|G\rangle_a\otimes|\lambda\rangle)\\
    &=(|G\rangle_a\otimes\mathds{1}_s)(\mathds{1}_a\otimes H|\lambda\rangle)=|G\rangle_a\otimes H|\lambda\rangle
\end{split}
\end{equation}
by \eqref{block_encoding_def}, which forms the first term in \eqref{lemma1step1}.

Since $|\lambda\rangle$ is an eigenstate of $H$, \eqref{lemma1step1} becomes
\begin{equation}
\label{lemma1step2}
\begin{split}
    U(|G\rangle_a\otimes|\lambda\rangle)&=|G\rangle_a\otimes \lambda|\lambda\rangle+\eta_\lambda|\perp_\lambda\rangle\\
    &=\lambda|G\rangle_a\otimes |\lambda\rangle+\sqrt{1-\lambda^2}|\perp_\lambda\rangle,
\end{split}
\end{equation}
where the second step follows by identifying ${\eta_\lambda=\sqrt{1-\lambda^2}}$ as the required normalization factor for the second term.
Hence, if we formally denote
\begin{equation}
\label{virtual_qubit}
    |G\rangle_a\otimes|\lambda\rangle=\begin{pmatrix}1\\0\end{pmatrix},\quad|\perp_\lambda\rangle=\begin{pmatrix}0\\1\end{pmatrix},
\end{equation}
we may write
\begin{equation}
\label{U_virtual_action}
    U\begin{pmatrix}1\\0\end{pmatrix}=\begin{pmatrix}\lambda\\\sqrt{1-\lambda^2}\end{pmatrix}.
\end{equation}
Note that the vectors in \eqref{virtual_qubit} represent the states of the ``virtual qubit'' that is referred to in the name \emph{qubitization}.

We also know that $U^2=\mathds{1}$ \eqref{self_inv}, which implies that within the two-dimensional subspace spanned by \eqref{virtual_qubit}, the matrix representing $U^2$ must be
\begin{equation}
    U^2\rightarrow\begin{pmatrix}1&0\\0&1\end{pmatrix}.
\end{equation}
From this and \eqref{U_virtual_action}, it follows that the matrix form of $U$ in this two-dimensional subspace must be
\begin{equation}
\label{U_matrix}
    U\rightarrow
    \begin{pmatrix}
        \lambda&\sqrt{1-\lambda^2}\\
        \sqrt{1-\lambda^2}&-\lambda
    \end{pmatrix},
\end{equation}
since \eqref{U_virtual_action} defines the first column of this matrix, and the second column is required to make it self-inverse.

The matrix for $U$ in \eqref{U_matrix} is a reflection.
We can transform it into a rotation if we can flip the signs of the entries in the second row, since the result will have the form
\begin{equation}
\label{rotation_matrix}
    U_y(\theta)=
    \begin{pmatrix}
        \cos\theta&\sin\theta\\
        -\sin\theta&\cos\theta
    \end{pmatrix}
\end{equation}
for a rotation angle $\theta$ defined by
\begin{equation}
\label{theta_def}
    \cos\theta=\lambda,\quad\sin\theta=\sqrt{1-\lambda^2}.
\end{equation}
In order to flip the signs in the second row of \eqref{U_matrix}, we need to multiply on the left by $Z$ within the two-dimensional subspace:
\begin{equation}
\begin{split}
    &Z
    \begin{pmatrix}
        \lambda&\sqrt{1-\lambda^2}\\
        \sqrt{1-\lambda^2}&-\lambda
    \end{pmatrix}\\
    &=
    \begin{pmatrix}
        1&0\\
        0&-1
    \end{pmatrix}
    \begin{pmatrix}
        \lambda&\sqrt{1-\lambda^2}\\
        \sqrt{1-\lambda^2}&-\lambda
    \end{pmatrix}\\
    &=
    \begin{pmatrix}
        \lambda&\sqrt{1-\lambda^2}\\
        -\sqrt{1-\lambda^2}&\lambda
    \end{pmatrix},
\end{split}
\end{equation}
which has the desired form \eqref{rotation_matrix}.

One way to implement the action of $Z$ on the two-dimensional subspace spanned by \eqref{virtual_qubit} is to reflect around $|G\rangle_a$ in the auxiliary space, i.e., to implement $R$ as defined in \eqref{R_def_app}.
The actions of $R$ on the subspace are
\begin{equation}
\begin{split}
    R\begin{pmatrix}1\\0\end{pmatrix}&=(2|G\rangle_a\langle G|_a-\mathds{1}_a)(|G\rangle_a\otimes|\lambda\rangle)\\
    &=|G\rangle_a\otimes|\lambda\rangle=\begin{pmatrix}1\\0\end{pmatrix},\\
    R\begin{pmatrix}0\\1\end{pmatrix}&=(2|G\rangle_a\langle G|_a-\mathds{1}_a)|\perp_\lambda\rangle\\
    &=-|\perp_\lambda\rangle=-\begin{pmatrix}0\\1\end{pmatrix},\\
\end{split}
\end{equation}
where the second equation follows from \eqref{perp_condition}.
Hence, the matrix representation of $R$ within the two-dimensional subspace is
\begin{equation}
    R\rightarrow\begin{pmatrix}1&0\\0&-1\end{pmatrix}=Z.
\end{equation}
Therefore, within the two-dimensional subspace
\begin{equation}
    RU\rightarrow
    \begin{pmatrix}
        \lambda&\sqrt{1-\lambda^2}\\
        -\sqrt{1-\lambda^2}&\lambda
    \end{pmatrix}
    =
    \begin{pmatrix}
        \cos\theta&\sin\theta\\
        -\sin\theta&\cos\theta
    \end{pmatrix},
\end{equation}
where $\theta$ is defined by \eqref{theta_def}.

As an aside, note that the form \eqref{R_def_app} for $R$ is not actually a necessary condition, although it is sufficient.
All that we strictly require is that $R$ implements a reflection around $|G\rangle_a$ within each two-dimensional subspace associated to each $\lambda$: its action on the remainder of the auxiliary space is irrelevant.
We will exploit this fact later on in constructing efficient implementations of $R$.

We can now use the fact that $k$ powers of a rotation by angle $\theta$ are simply the rotation by angle $k\theta$, so
\begin{equation}
    (RU)^k\rightarrow
    \begin{pmatrix}
        \cos k\theta&\sin k\theta\\
        -\sin k\theta&\cos k\theta
    \end{pmatrix}.
\end{equation}
But
\begin{equation}
    \cos\theta=\lambda\quad\Rightarrow\quad\cos k\theta=T_k(\lambda),
\end{equation}
where $T_k(\cdot)$ is the $k$th Chebyshev polynomial of the first kind, so
\begin{equation}
    (RU)^k\rightarrow
    \begin{pmatrix}
        T_k(\lambda)&\sqrt{1-T_k^2(\lambda)}\\
        -\sqrt{1-T_k^2(\lambda)}&T_k(\lambda)
    \end{pmatrix}.
\end{equation}
By the definition \eqref{virtual_qubit} of the two-dimensional subspace associated to $\lambda$, this implies that
\begin{equation}
\begin{split}
    (\langle G|_a\otimes\mathds{1}_s)(RU)^k(|G\rangle_a\otimes|\lambda\rangle)&=T_k(\lambda)|\lambda\rangle\\
    &=T_k(H)|\lambda\rangle.
\end{split}
\end{equation}
The final step in the proof is to note that any arbitrary state $|\psi\rangle$ of the system can be expressed as a superposition of Hamiltonian eigenstates:
\begin{equation}
    |\psi\rangle=\sum_\lambda\beta_\lambda|\lambda\rangle,
\end{equation}
so
\begin{equation}
\begin{split}
    &(\langle G|_a\otimes\mathds{1}_s)(RU)^k(|G\rangle_a\otimes\mathds{1}_s)\,|\psi\rangle\\
    &=\sum_\lambda\beta_\lambda(\langle G|_a\otimes\mathds{1}_s)(RU)^k(|G\rangle_a\otimes|\lambda\rangle)\\
    &=\sum_\lambda\beta_\lambda T_k(H)|\lambda\rangle=T_k(H)|\psi\rangle.
\end{split}
\end{equation}
Since this applies for any state $|\psi\rangle$, we have
\begin{equation}
    (\langle G|_a\otimes\mathds{1}_s)(RU)^k(|G\rangle_a\otimes\mathds{1}_s)=T_k(H),
\end{equation}
as desired.

\end{proof}

\section{Convergence of the noisy, thresholded Krylov method with Chebyshev polynomials}
\label[appendix]{lanczos_convergence_app}

In this section we prove \cref{error_theorem}, which is stated in the main text and bounds the error of our method subject to noise and regularization via eigenvalue thresholding.

\textbf{Theorem~\ref{error_theorem}.}
\emph{
    Let $H$ be an $N\times N$ Hamiltonian with energies ${E_0\le E_1\le\cdots\le E_{N-1}}$ in the range $[-1,1]$.
    Suppose we compute the $D$-dimensional Krylov space spanned by \eqref{chebyshev_lanczos} using the method in \cref{method}, yielding noisy estimates
    \begin{equation}
        (\widetilde{\textbf{H}},\widetilde{\textbf{S}})=(\textbf{H}+\boldsymbol\Delta_H,\textbf{S}+\boldsymbol\Delta_S)
    \end{equation}
    of $(\textbf{H},\textbf{S})$, for some Hermitian perturbations $(\boldsymbol\Delta_H,\boldsymbol\Delta_S)$ whose spectral norms are $(\eta_H,\eta_S)$.
    Denote the total noise rate as
    \begin{equation}
        \eta\coloneqq\sqrt{\eta_H^2+\eta_S^2}.
    \end{equation}
    Let
    \begin{equation}
        \gamma_0=\langle E_0|\psi_0\rangle
    \end{equation}
    be the overlap of the initial reference state $|\psi_0\rangle$ with the true ground state $|E_0\rangle$.
    Let $\widetilde{\textbf{S}}$ be regularized by projecting it (and $\widetilde{\textbf{H}}$) onto the subspace spanned by its eigenvectors with eigenvalues above some threshold $\epsilon>0$.
    Let $\epsilon_\text{total}$ be the sum of the eigenvalues of $\textbf{S}$ discarded by regularizing it according to $\epsilon$ in the same way.
    Then provided the noise rate $\eta$ is sufficiently small (small enough that the assumptions of \cref{noise_theorem} can be satisfied), there exists a choice of threshold
    \begin{equation}
    \label{threshold_condition_app}
        \epsilon=\Omega\left((D^4\eta)^\frac{1}{1+\alpha}\right)
    \end{equation}
    for some constant $0\le\alpha\le1/2$ such that for any $\delta>0$, the error $\mathcal{E}$ in the ground state energy estimate coming from the regularized version of the noisy problem $(\widetilde{\textbf{H}},\widetilde{\textbf{S}})$ is bounded by
    \begin{equation}
    \label{total_error_bound_app}
    \begin{split}
        \mathcal{E}\le O\Bigg(&(D^4\eta)^{\frac{1}{1+\alpha}}+\frac{\sqrt{\delta}\,\epsilon_\text{total}}{|\gamma_0|^2}\\
        &\quad+\delta+\frac{1}{|\gamma_0|^2}\left(1+\frac{\delta}{2}\right)^{-D}\Bigg).
    \end{split}
    \end{equation}
}

\noindent
Note: the first term in \eqref{total_error_bound_app} comes from Theorem 2.7 in \cite{epperly2021subspacediagonalization} (reproduced below as \cref{noise_theorem}), and is subject to additional assumptions given in the statement of \cref{noise_theorem} below.

\begin{proof}

Let $E_0$ be the true ground state energy of $H$, let $E_0'$ be the energy estimate from the noiseless, thresholded Krylov space, and let $\widetilde{E}_0'$ be the energy estimate from the noisy, thresholded Krylov space.
Then the total error is $\mathcal{E}=|\widetilde{E}_0'-E_0|$, for which we obtain the bound \eqref{total_error_bound_app} via the triangle inequality as a first step:
\begin{equation}
\label{first_step_bound}
    \mathcal{E}=\left|\widetilde{E}_0'-E_0\right|\le\left|\widetilde{E}_0'-E_0'\right|+\left|E_0'-E_0\right|.
\end{equation}
We will upper bound the first term, characterizing the error due to noise, using Theorem 2.7 of \cite{epperly2021subspacediagonalization}, which can be applied directly to our case.
We will upper bound the second term in \eqref{first_step_bound} using \cref{threshold_theorem}, proved below, which is based on Theorem 3.1 in \cite{epperly2021subspacediagonalization}, but which we had to modify to apply to our case, since the version in \cite{epperly2021subspacediagonalization} is specific to subspaces generated by real time evolutions.

Theorem 2.7 in~\cite{epperly2021subspacediagonalization} may be restated in our notation as follows:
\begin{lemma}[Theorem 2.7 in~\cite{epperly2021subspacediagonalization}]
\label{noise_theorem}
    Instate the prevailing notation.
    Let $(\lambda_i,\textbf{v}_i)$ be the eigenpairs of $\textbf{S}$, ordered by eigenvalue from largest to smallest.
    Let $m$ be the greatest index such that $\lambda_m>\epsilon$.
    Four assumptions are required for this theorem to hold:
    \begin{enumerate}
        \item (gap at thresholding level) Assume that the eigenvalues of $\textbf{S}$ are sufficiently separated at the threshold level and the noise rate is low enough that
        \begin{equation}
        \label{eig_separation_condition}
            \lambda_{m+1}+\eta_S\le\epsilon.
        \end{equation}
        This assumption is merely to ensure that during thresholding the same numbers of eigenvalues are discarded from the noisy and noiseless overlap matrices, so that the resulting thresholded matrix problems have the same dimensions.

        \item (bound on $\textbf{H}$ relative to $\textbf{S}$) Assume that
        \begin{equation}
            \left|\textbf{v}_i^\dagger\textbf{H}\textbf{v}_j\right|\le\mu\min(\lambda_i,\lambda_j)^{1-\alpha}\max(\lambda_i,\lambda_j)^\alpha
        \end{equation}
        for all $i,j$, for some constants $\mu>0$ and ${0\le\alpha\le1/2}$.
        This bound always holds if ${\alpha=1/2}$, but smaller $\alpha$ may be more realistic in practice (see \cite{epperly2021subspacediagonalization}).

        \item (bound on noise rate relative to threshold) Let $\rho>0$ be such that $(1+\rho)\epsilon\le\lambda_m$.
        Assume that
        \begin{equation}
        \label{epsilon_lower_bound}
            D^4\chi\le\epsilon,
        \end{equation}
        where $\chi$ is defined as
        \begin{equation}
        \label{chi_def}
            \chi\coloneqq3(2+\mu)\left(1+\frac{1}{\rho}\right)\left(\frac{\|\textbf{S}\|}{\epsilon}\right)^\alpha\eta_S+\eta_H.
        \end{equation}

        \item (gap condition) Assume that 
        \begin{equation}
        \label{gap_condition}
            \arctan(E_1')-\arctan(E_0')\ge\arcsin\left(\frac{D^4\chi}{\epsilon}\right),
        \end{equation}
        where $E_0',E_1'$ are the lowest and second-to-lowest eigenvalues of the noiseless, thresholded problem.
    \end{enumerate}
    Then the difference between the lowest eigenvalue $E_0'$ of the noiseless thresholded problem and the lowest eigenvalue $\widetilde{E}_0'$ of the noisy thresholded problem is bounded as
    \begin{equation}
    \label{noisy_energy_bound}
        \left|\arctan(E_0')-\arctan(\widetilde{E}_0')\right|\le\arcsin(\kappa D^4\chi),
    \end{equation}
    ~\\
    where $\kappa$ is the condition number of $\arctan(E_0')$.
    [Note: Theorem 2.7 in~\cite{epperly2021subspacediagonalization} also includes the assumption
    \begin{equation}
        \left(1+\frac{1}{\rho}\right)\eta_S\le\epsilon,
    \end{equation}
    but this is guaranteed to hold by \eqref{epsilon_lower_bound}.]
\end{lemma}

Our input Hamiltonian will be normalized such that its eigenvalues lie inside the range $[-1,1]$.
Therefore, the condition number $\kappa$ of $\arctan(E_0')$ is bounded as
\begin{widetext}
\begin{equation}
\begin{split}
    \min_{x\in[-1,1]}\left|\frac{x\arctan'(x)}{\arctan(x)}\right|&\le\kappa\le\max_{x\in[-1,1]}\left|\frac{x\arctan'(x)}{\arctan(x)}\right|\\
    \Rightarrow\quad\min_{x\in[-1,1]}\left|\frac{x}{(1+x^2)\arctan(x)}\right|&\le\kappa\le\max_{x\in[-1,1]}\left|\frac{x}{(1+x^2)\arctan(x)}\right|\\
    \Rightarrow\quad\left|\frac{1}{2\arctan(1)}\right|&\le\kappa\le\lim_{x\rightarrow0}\frac{x}{(1+x^2)\arctan(x)}=\lim_{x\rightarrow0}\frac{x}{\arctan(x)}\\
    \Rightarrow\quad\frac{2}{\pi}&\le\kappa\le1.
\end{split}
\end{equation}
\end{widetext}
Hence we can replace $\kappa$ by $1$ in \eqref{noisy_energy_bound}, and the resulting bound will be weaker but equivalent in scaling:
\begin{equation}
\label{noisy_energy_bound_2}
    \left|\arctan(E_0')-\arctan(\widetilde{E}_0')\right|\le\arcsin(D^4\chi),
\end{equation}
Also, for $x\in[-1,1]$, the minimum value of the first derivative of $\arctan(x)$ is $1/2$, so
\begin{equation}
    \frac{1}{2}|E_0'-\widetilde{E}_0'|\le\left|\arctan(E_0')-\arctan(\widetilde{E}_0')\right|.
\end{equation}
Finally, for $x\in[0,1]$, $\arcsin(x)\le\pi x/2$, so we can obtain a simplified energy bound of
\begin{equation}
\label{noisy_energy_bound_3}
    |E_0'-\widetilde{E}_0'|\le\pi D^4\chi.
\end{equation}

Similarly, since the maximum value of the first derivative of $\arctan(x)$ on $[-1,1]$ is $1$ and $\arcsin(x)\ge x$ on $[0,1]$, the condition \eqref{gap_condition} on the gap can be simplified to
\begin{equation}
\label{gap_condition_2}
    \Delta'\coloneqq E_1'-E_0'\ge\frac{D^4\chi}{\epsilon}.
\end{equation}
Hence we just need to choose
\begin{equation}
\label{epsilon_condition_1}
    \epsilon\ge\frac{D^4\chi}{\Delta'}
\end{equation}
and \eqref{epsilon_lower_bound} will also be guaranteed, so the only remaining assumption in the statement of \cref{noise_theorem} will be that \eqref{eig_separation_condition} holds.

Inserting the definition \eqref{chi_def} of $\chi$ into \eqref{epsilon_condition_1} yields
\begin{equation}
\label{epsilon_condition_2}
    \epsilon=\Theta\left(\left(\frac{D^4\eta}{\Delta'}\right)^{\frac{1}{1+\alpha}}\right).
\end{equation}
If we now insert the definition of $\chi$ into \eqref{noisy_energy_bound_3} and simplify within asymptotic notation, we can obtain
\begin{equation}
\label{noisy_energy_bound_4}
    |E_0'-\widetilde{E}_0'|\le O\left(\frac{D^4\eta}{\epsilon^\alpha}\right)=O\left((D^4\eta)^{\frac{1}{1+\alpha}}(\Delta')^{\frac{\alpha}{1+\alpha}}\right),
\end{equation}
where the second step uses \eqref{epsilon_condition_2}.
Since $\Delta'$ will be smaller than 1 except possibly at the lowest Krylov space dimensions, we can obtain a weaker upper bound by simply dropping the factor that depends on $\Delta'$, obtaining
\begin{equation}
\label{noisy_energy_bound_5}
    |E_0'-\widetilde{E}_0'|\le O\left((D^4\eta)^{\frac{1}{1+\alpha}}\right);
\end{equation}
in practice, the inclusion of $\Delta'$ in the denominator of the lower bound to $\epsilon$ as in \eqref{epsilon_condition_2} is probably unrealistic anyway, so we are likely not sacrificing any tightness of the bound by making this simplification.
Inserting \eqref{noisy_energy_bound_5} in \eqref{first_step_bound} provides the first term in \eqref{total_error_bound_app}, while \eqref{epsilon_condition_2} yields the condition \eqref{threshold_condition_app}.

The remaining terms in \eqref{total_error_bound_app}, coming from the second term in \eqref{first_step_bound}, are provided by \eqref{thresholded_energy_bound}, which is the result of \cref{threshold_theorem}, below.
When inserting this into \eqref{first_step_bound}, we assume that the threshold $\epsilon$ is small enough that the denominator in \eqref{thresholded_energy_bound} scales as $\Theta(|\gamma_0|^2)$, allowing us to ignore the second term the denominator in the asymptotic scaling.
This is why the statement of \cref{error_theorem} only claims that if the noise rate is sufficiently small then a threshold exists such that \eqref{total_error_bound_app} is satisfied.

\end{proof}

\setcounter{theorem}{1}
\begin{theorem}
\label{threshold_theorem}
    Instate the prevailing notation.
    The error in the noiseless thresholded ground state energy estimate $E_0'$ is bounded as
    \begin{equation}
    \label{thresholded_energy_bound}
    \begin{split}
        0&\le E_0'-E_0\\
        &\le\delta+8\frac{\sqrt{\delta}\,\epsilon_\text{total}+\left(1-\gamma^2+4\epsilon_\text{total}\right)\left(1+\frac{\delta}{2}\right)^{-2\lfloor k/2\rfloor}}{\left(|\gamma_0|-2\sqrt{(k+1)\epsilon}\right)^2},
    \end{split}
    \end{equation}
    where
    \begin{equation}
        \gamma\coloneqq\sqrt{\sum_\text{$i$ s.t. $E_i-E_0\le\delta$}|\gamma_i|^2}
    \end{equation}
    is the overlap of the initial state with the subspace whose energy is within $\delta$ of $E_0$.
    [Note: this form of the bound only holds when ${|\gamma_0|>\sqrt{2D}\epsilon}$, and is only tight when the denominator in \eqref{thresholded_energy_bound} is $\Omega(|\gamma_0|^2)$.]
\end{theorem}
\begin{proof}

Let $\textbf{K}$ be the \emph{Krylov matrix} whose columns are the Krylov vectors, i.e.,
\begin{equation}
    \textbf{K}=\Big[|\psi_0\rangle\quad T_1(H)|\psi_0\rangle\quad T_2(H)|\psi_0\rangle~\cdots~T_k(H)|\psi_0\rangle\Big],
\end{equation}
where
\begin{equation}
    k\coloneqq D-1.
\end{equation}
The projected Hamiltonian $\textbf{H}$ and overlap matrix $\textbf{S}$ are given by
\begin{equation}
    \textbf{H}=\textbf{K}^\dagger H\textbf{K},\quad\textbf{S}=\textbf{K}^\dagger\textbf{K}.
\end{equation}
Hence, the eigenvalues of $\textbf{S}$ are the squares of the singular values of $\textbf{K}$, and the corresponding eigenvectors of $\textbf{S}$ are the right singular vectors of $\textbf{K}$.
Therefore, applying thresholding that removes the eigenspaces of $\textbf{S}$ with eigenvalues smaller than $\epsilon$ is equivalent to truncating the Krylov space to the span of the left singular vectors (which are states in the Hilbert space) of $\textbf{K}$ with singular values at least $\sqrt{\epsilon}$.

We can factor $\textbf{K}$ as follows:
\begin{widetext}
\begin{equation}
    \textbf{K}=\underbrace{\Big[|E_0\rangle\quad |E_1\rangle\quad |E_2\rangle~\cdots~|E_{N-1}\rangle\Big]}_{\boldsymbol\Psi}
    \underbrace{
    \begin{bmatrix}
        \gamma_0&&&& \\
        & \gamma_1 &&& \\
        && \gamma_2 && \\
        &&& \ddots & \\
        &&&& \gamma_{N-1} \\
    \end{bmatrix}
    }_{\boldsymbol\Gamma}
    \underbrace{
    \begin{bmatrix}
        1 & T_1(E_0) & \cdots & T_k(E_0) \\[1mm]
        1 & T_1(E_1) & \cdots & T_k(E_1) \\[1mm]
        1 & T_1(E_2) & \cdots & T_k(E_2) \\\
        \vdots & \vdots & \ddots & \vdots \\[0.6mm]
        1 & T_1(E_{N-1}) & \cdots & T_k(E_{N-1})
    \end{bmatrix}
    }_{\textbf{F}}.
\end{equation}
\end{widetext}
Let $\textbf{K}'=\boldsymbol\Psi\boldsymbol\Gamma\textbf{F}'$ denote the matrix obtained by setting all singular values of $\textbf{K}$ smaller than $\sqrt{\epsilon}$ to zero.

Let $R'$ be the range (column space) of $\textbf{K}'$, i.e., the thresholded Krylov space, spanned by the left singular vectors of $\textbf{K}$ with singular values at least $\sqrt{\epsilon}$.
Then by the Rayleigh-Ritz variational principle, the error in the ground state energy estimate in the thresholded Krylov space
\begin{equation}
\label{rayleigh_ritz}
    E_0'-E_0=\min_{|\psi\rangle\in R'\setminus\{\textbf{0}\}}\frac{\langle\psi|(H-E_0)|\psi\rangle}{\langle\psi|\psi\rangle}.
\end{equation}
We will upper bound the above minimization by inserting a particular state $|\psi\rangle\in R'$ into its argument.

To construct this state, we use the polynomial $p^*$ given in \cref{poly_approx_lemma}, below, with the parameters $a,b,d$ in \cref{poly_approx_lemma} set to $a=\delta^2/4$, $b=1$, and $d=\lfloor k/2\rfloor$: in terms of this, define the function $f$ as
\begin{equation}
    f(x)\coloneqq p^*\left(\frac{(x-E_0)^2}{4}\right).
\end{equation}
By \cref{poly_approx_lemma},
\begin{equation}
\label{poly_zero_val}
    f(E_0)=p^*(0)=1,
\end{equation}
but for any $x\in[-1,1]$ such that ${|x-E_0|\ge\delta}$,
\begin{equation}
\label{poly_bound}
    |f(x)|\le2\left(1+\sqrt{\frac{a}{b}}\right)^{-d}=2\left(1+\frac{\delta}{2}\right)^{-\lfloor k/2\rfloor}.
\end{equation}
Note that this bound is only tight when it is much less than $1$.
Since $p^*$ has degree $\lfloor k/2\rfloor$, $f(x)$ has degree either $k-1$ or $k$.
Let $c_j$ be the coefficients of the Chebyshev expansion of $f$:
\begin{equation}
\label{poly_def}
    f(x)=\sum_{j=0}^{k}c_jT_j(x).
\end{equation}
Since ${|p^*(x)|\le1}$ for ${0\le x\le 1}$, ${|f(x)|\le1}$ for ${-1\le x\le1}$, and for $x$ outside the interval ${[E_0-\delta,E_0+\delta]}$, the much tighter bound \eqref{poly_bound} holds.
Therefore,
\begin{equation}
\begin{split}
    &\int_{-1}^1dx\frac{f^2(x)}{\sqrt{1-x^2}}\\
    &\quad\le\int_{E_0-\delta}^{E_0+\delta}dx\frac{1}{\sqrt{1-x^2}}\\
    &\qquad+\int_{[-1,1]\setminus[E_0-\delta,E_0+\delta]}dx\frac{4\left(1+\frac{\delta}{2}\right)^{-2\lfloor k/2\rfloor}}{\sqrt{1-x^2}}\\
    &\quad\le\int_{1-2\delta}^{1}dx\frac{1}{\sqrt{1-x^2}}+\int_{-1}^{1-2\delta}dx\frac{4\left(1+\frac{\delta}{2}\right)^{-2\lfloor k/2\rfloor}}{\sqrt{1-x^2}}\\
    &\quad\le\cos^{-1}(1-2\delta)\\
    &\qquad+4\left(\pi-\cos^{-1}(1-2\delta)\right)\left(1+\frac{\delta}{2}\right)^{-2\lfloor k/2\rfloor}\\
    &\quad\le\pi\sqrt{\delta}+4(\pi-2\sqrt{\delta})\left(1+\frac{\delta}{2}\right)^{-2\lfloor k/2\rfloor},
\end{split}
\end{equation}
which holds for $\delta<1/4$, so since
\begin{equation}
    \int_{-1}^1dx\frac{T_i(x)T_j(x)}{\sqrt{1-x^2}}=
    \begin{cases}
        \pi\quad\text{if $i=j=0$},\\
        \pi/2\quad\text{if $i=j\neq0$},\\
        0\quad\text{if $i\neq j$},
    \end{cases}
\end{equation}
we obtain the following bound on the coefficients $c_j$:
\begin{equation}
\label{coefficient_bound}
\begin{split}
    \sum_{j=0}^kc_j^2&\le c_0^2+\sum_{j=0}^kc_j^2=\frac{2}{\pi}\int_{-1}^1dx\frac{f^2(x)}{\sqrt{1-x^2}}\\
    &\le2\sqrt{\delta}+8\left(1-\frac{2}{\pi}\sqrt{\delta}\right)\left(1+\frac{\delta}{2}\right)^{-2\lfloor k/2\rfloor}\\
    &\coloneqq g(\delta,k).
\end{split}
\end{equation}
The above is really just an application of Parseval's Theorem, using the fact that the Chebyshev polynomials are orthogonal on $[-1,1]$ with respect to the measure ${dx/\sqrt{1-x^2}}$.
As noted after \eqref{poly_bound}, this upper bound $g(\delta,k)$ is only tight when its second term is much smaller than $1$, which will always hold in our regime of interest, since as we will see, this second term is proportional to the main energy error term we will end up with.

Let $|\psi\rangle$ be defined in terms of the coefficients $c_j$ as
\begin{equation}
    |\psi\rangle=\sum_{j=0}^{k}c_j\sum_{i=0}^{N-1}\gamma_i\textbf{F}'_{ij}|E_i\rangle\in R'
\end{equation}
(note that this state is in $R'$ because $\sum_{i=0}^{N-1}\gamma_i\textbf{F}'_{ij}|E_i\rangle$ is the $j$th column in $\textbf{K}'$.)
Inserting this into \eqref{rayleigh_ritz} yields
\begin{equation}
\label{energy_bound_1}
    E_0'-E_0\le\frac{\sum_{i=1}^{N-1}|\gamma_i|^2\left|\sum_{j=0}^{k}c_j\textbf{F}'_{ij}\right|^2(E_i-E_0)}{\sum_{i=0}^{N-1}|\gamma_i|^2\left|\sum_{j=0}^{k}c_j\textbf{F}'_{ij}\right|^2}.
\end{equation}

Let $I$ be the least value of $i$ such that $E_I-E_0>\delta$.
We separately upper bound the parts of \eqref{energy_bound_1} coming from terms in the numerator with $i\ge I$ and $i<I$.
We start with the former, i.e., let $i\ge I$ and consider the term 
\begin{equation}
\label{term_outside_subspace}
\begin{split}
    &\frac{|\gamma_i|^2\left|\sum_{j=0}^{k}c_j\textbf{F}'_{ij}\right|^2(E_i-E_0)}{\sum_{l=0}^{N-1}|\gamma_l|^2\left|\sum_{j=0}^{k}c_j\textbf{F}'_{lj}\right|^2}\\
    &\le\frac{|\gamma_i|^2\left|\sum_{j=0}^{k}c_j\textbf{F}'_{ij}\right|^2(E_i-E_0)}{|\gamma_0|^2\left|\sum_{j=0}^{k}c_j\textbf{F}'_{0j}\right|^2}.
\end{split}
\end{equation}
We upper bound the numerator as follows, using \eqref{poly_bound} and \eqref{coefficient_bound}:
\begin{equation}
\label{numer_upper_1}
\begin{split}
    &\left|\sum_{j=0}^{k}c_j\textbf{F}'_{ij}\right|^2=\left|\sum_{j=0}^{k}c_j(\textbf{F}'_{ij}-T_j(E_i))+\sum_{j=0}^{k}c_jT_j(E_i)\right|^2\\
    &=\left|\sum_{j=0}^{k}c_j(\textbf{F}'_{ij}-T_j(E_i))+f(E_i)\right|^2\\
    &\le2\left(\sum_{j=0}^{k}|c_j|^2\right)\left(\sum_{j=0}^{k}|\textbf{F}'_{ij}-T_j(E_i)|^2\right)+2\left|f(E_i)\right|^2\\
    &\le2g(\delta,k)\left(\sum_{j=0}^{k}\alpha_{ij}^2\right)+4\left(1+\frac{\delta}{2}\right)^{-2\lfloor k/2\rfloor},
\end{split}
\end{equation}
where in the last line we defined
\begin{equation}
\label{alpha_def}
    \alpha_{ij}\coloneqq|\textbf{F}'_{ij}-T_j(E_i)|.
\end{equation}
For the denominator of this term, we use
\begin{equation}
\label{denom_lower_1}
\begin{split}
    \left|\sum_{j=0}^{k}c_j\textbf{F}'_{0j}\right|&=\left|\sum_{j=0}^{k}c_jT_j(E_0)+\sum_{j=0}^{k}c_j(\textbf{F}'_{0j}-T_j(E_0))\right|\\
    &\ge\left|f(E_0)\right|-\left|\sum_{j=0}^{k}c_j(\textbf{F}'_{0j}-T_j(E_0))\right|\\
    &\ge1-\sum_{j=0}^{k}|c_j|\left|\textbf{F}'_{0j}-T_j(E_0)\right|\\
    &\ge1-2\sum_{j=0}^{k}\alpha_{0j},
\end{split}
\end{equation}
where the second line follows from the reverse triangle inequality and \eqref{poly_def}, the third line follows from the triangle inequality and \eqref{poly_zero_val}, and the last line follows from the fact that each $|c_j|\le2$ (obtained by evaluating the integral in the first line of \eqref{coefficient_bound} replacing $f^2(x)$ with its upper bound of $1$).
Extending the above bound,
\begin{equation}
\label{denom_lower}
\begin{split}
    \left|\sum_{j=0}^{k}c_j\textbf{F}'_{0j}\right|&\ge1-2\sqrt{(k+1)\sum_{j=0}^{k}\alpha_{0j}^2}\\
    &\ge1-\frac{2}{|\gamma_0|}\sqrt{(k+1)\epsilon},
\end{split}
\end{equation}
using the spectral norm bound
\begin{equation}
\label{spectral_norm_bound}
\begin{split}
    \sqrt{\epsilon}&\ge\|\textbf{K}'-\textbf{K}\|=\|\boldsymbol\Gamma(\textbf{F}'-\textbf{F})\|\\
    &\ge\sqrt{|\gamma_0|^2\sum_{j=0}^k|\textbf{F}'_{0j}-T_j(E_0)|^2}=|\gamma_0|\sqrt{\sum_{j=0}^k\alpha_{0j}^2}.
\end{split}
\end{equation}
Hence, the term \eqref{term_outside_subspace} is upper bounded by
\begin{equation}
\label{term_outside_subspace_upper_bound}
\begin{split}
    &\frac{2|\gamma_i|^2(E_i-E_0)\left(g(\delta,k)\sum_{j=0}^{k}\alpha_{ij}^2+2\left(1+\frac{\delta}{2}\right)^{-2\lfloor k/2\rfloor}\right)}{\left(|\gamma_0|-2\sqrt{(k+1)\epsilon}\right)^2}\\
    &\le\frac{4|\gamma_i|^2\left(g(\delta,k)\sum_{j=0}^{k}\alpha_{ij}^2+2\left(1+\frac{\delta}{2}\right)^{-2\lfloor k/2\rfloor}\right)}{\left(|\gamma_0|-2\sqrt{(k+1)\epsilon}\right)^2}.
\end{split}
\end{equation}

Now consider together all of the terms in \eqref{energy_bound_1} with $i<I$:
\begin{equation}
\label{terms_inside_subspace}
\begin{split}
    &\frac{\sum_{i=0}^{I-1}|\gamma_i|^2\left|\sum_{j=0}^{k}c_j\textbf{F}'_{ij}\right|^2(E_i-E_0)}{\sum_{l=0}^{N-1}|\gamma_l|^2\left|\sum_{j=0}^{k}c_j\textbf{F}'_{lj}\right|^2}\\
    &\le\frac{\sum_{i=0}^{I-1}|\gamma_i|^2\left|\sum_{j=0}^{k}c_j\textbf{F}'_{ij}\right|^2(E_i-E_0)}{\sum_{i=0}^{I-1}|\gamma_i|^2\left|\sum_{j=0}^{k}c_j\textbf{F}'_{ij}\right|^2}.
\end{split}
\end{equation}
The right-hand side of \eqref{terms_inside_subspace} is an expectation value of the difference between energy and ground state energy for some state in the subspace whose energies are within $\delta$ of the ground state energy.
Hence by the Rayleigh-Ritz variational principle, \eqref{terms_inside_subspace} is upper bounded by $\delta$:
\begin{equation}
\label{terms_inside_subspace_upper_bound}
    \frac{\sum_{i=0}^{I-1}|\gamma_i|^2\left|\sum_{j=0}^{k}c_j\textbf{F}'_{ij}\right|^2(E_i-E_0)}{\sum_{l=0}^{N-1}|\gamma_l|^2\left|\sum_{j=0}^{k}c_j\textbf{F}'_{lj}\right|^2}\le\delta.
\end{equation}

Inserting \eqref{term_outside_subspace_upper_bound} and \eqref{terms_inside_subspace_upper_bound} into \eqref{energy_bound_1}, we obtain
\begin{equation}
\label{energy_bound_2}
\begin{split}
    &E_0'-E_0\\
    &\le\delta+\sum_{i=I}^{N-1}\frac{4|\gamma_i|^2\left(g(\delta,k)\sum_{j=0}^{k}\alpha_{ij}^2+2\left(1+\frac{\delta}{2}\right)^{-2\lfloor k/2\rfloor}\right)}{\left(|\gamma_0|-2\sqrt{(k+1)\epsilon}\right)^2}\\
    &=\delta+4\frac{g(\delta,k)\epsilon_\text{total}+2(1-\gamma^2)\left(1+\frac{\delta}{2}\right)^{-2\lfloor k/2\rfloor}}{\left(|\gamma_0|-2\sqrt{(k+1)\epsilon}\right)^2}\\
    &\le\delta+8\frac{\sqrt{\delta}\,\epsilon_\text{total}+\left(1-\gamma^2+4\epsilon_\text{total}\right)\left(1+\frac{\delta}{2}\right)^{-2\lfloor k/2\rfloor}}{\left(|\gamma_0|-2\sqrt{(k+1)\epsilon}\right)^2},
\end{split}
\end{equation}
our desired result, where the last line is obtained by inserting the definition of $g(\delta,k)$ as in \eqref{coefficient_bound} and upper bounding $1-\frac{2}{\pi}\sqrt{\delta}\le1$, and in the second line we inserted 
\begin{equation}
    \gamma\coloneqq\sqrt{\sum_{i=0}^{I-1}|\gamma_i|^2}=\sqrt{\sum_\text{$i$ s.t. $E_i-E_0\le\delta$}|\gamma_i|^2}
\end{equation}
and $\epsilon_\text{total}$, which we define to be the squared Frobenius distance from $\textbf{K}$ to $\textbf{K}'$:
\begin{equation}
\label{eps_total}
    \epsilon_\text{total}\coloneqq\|\textbf{K}'-\textbf{K}\|^2_F=\|\boldsymbol\Gamma(\textbf{F}'-\textbf{F})\|^2_F=\sum_{i=0}^{N-1}\sum_{j=0}^{k}|\gamma_i|^2\alpha_{ij}^2
\end{equation}
(the second equality follows because $\boldsymbol\Psi$ is unitary.)
Since $\textbf{K}'-\textbf{K}$ is a matrix whose only nonzero singular values are those singular values of $\textbf{K}$ that are smaller than $\sqrt{\epsilon}$ and are thus dropped in thresholding, and the Frobenius norm may also be written
\begin{equation}
    \epsilon_\text{total}=\|\textbf{K}'-\textbf{K}\|^2_F=\text{Tr}\Big((\textbf{K}'-\textbf{K})^\dagger(\textbf{K}'-\textbf{K})\Big),
\end{equation}
$\epsilon_\text{total}$ as defined by \eqref{eps_total} is equal to the sum of squares of the singular values of $\textbf{K}$ that are discarded in thresholding.
Since $\textbf{S}=\textbf{K}^\dagger\textbf{K}$, $\epsilon_\text{total}$ is the sum of eigenvalues of $\textbf{S}$ that are discarded in thresholding, as claimed in the theorem statement.

\end{proof}

\begin{lemma}[\cite{bjorck2014numerical},Theorem 4.1.11]
\label{poly_approx_lemma}
    Let $0<a<b$ and let $\Pi_d^*$ be the space of residual polynomials (polynomials whose value at $0$ is $1$) of degree at most $d$.
    The solution to
    \begin{equation}
        \beta(a,b,d)=\min_{p\in\Pi_d^*}\max_{x\in[a,b]}|p(x)|
    \end{equation}
    is
    \begin{equation}
    \label{optimal_poly}
        p^*(x)=\frac{T_d\left(\frac{b+a-2x}{b-a}\right)}{T_d\left(\frac{b+a}{b-a}\right)},
    \end{equation}
    and the corresponding minimal value is
    \begin{equation}
    \label{optimal_val}
    \begin{split}
        \beta(a,b,d)&=T_d^{-1}\left(\frac{b+a}{b-a}\right)\\
        &\le2\left(1+\sqrt{\frac{a}{b}}\right)^{-d}.
    \end{split}
    \end{equation}
    Note that the upper bound is only tight when it is much less than $1$.
\end{lemma}
\begin{proof}
    Eq.~\eqref{optimal_poly} is the result given in Theorem 4.1.11 in~\cite{bjorck2014numerical}.
    To obtain \eqref{optimal_val}, note that for $x\in[a,b]$,
    \begin{equation}
    \label{chebyshev_domain}
        -1\le\frac{b+a-2x}{b-a}\le1,
    \end{equation}
    so
    \begin{equation}
        \left|T_d\left(\frac{b+a-2x}{b-a}\right)\right|\le1
    \end{equation}
    with equality at either end of the domain \eqref{chebyshev_domain}.
    Hence
    \begin{equation}
    \begin{split}
        &\beta(a,b,d)=\min_{p\in\Pi_k^*}\max_{x\in[a,b]}|p(x)|\\
        &=\max_{x\in[a,b]}\left|\frac{T_d\left(\frac{b+a-2x}{b-a}\right)}{T_d\left(\frac{b+a}{b-a}\right)}\right|=T_d^{-1}\left(\frac{b+a}{b-a}\right)\\
        &\le2\left(\frac{b+a}{b-a}+\sqrt{\left(\frac{b+a}{b-a}\right)^2-1}\right)^{-d}\\
        &=2\left(\frac{b-a}{b+a}\right)^d\left(1+\sqrt{1-\left(\frac{b-a}{b+a}\right)^2}\right)^{-d}\\
        &\le2\left(1+\sqrt{1-\left(\frac{b-a}{b+a}\right)^2}\right)^{-d}\\
        &\le2\left(1+\sqrt{\frac{a}{b}}\right)^{-d},
    \end{split}
    \end{equation}
    where the last step follows after some straightforward but tedious algebra.
\end{proof}

\section{Block encoding based on binary representation of Pauli operators}
\label[appendix]{block_encoding_app}

\subsection{Block encoding components}

Recall that to implement a block encoding, we need a preparation procedure for the state
\begin{equation}
    |G\rangle_a=\sum_{i=0}^{T-1}\sqrt{\alpha_i}|i\rangle_a
\end{equation}
and an implementation of the unitary
\begin{equation}
\label{unitary_app}
    U=\sum_{i=0}^{T-1}|i\rangle_a\langle i|_a\otimes P_i.
\end{equation}
Suppose we let
\begin{equation}
    |i\rangle_a=|\vec{x},\vec{z}\rangle_a,
\end{equation}
where $(\vec{x},\vec{z})$ is the binary representation of $P_i$~\cite{gottesman97a,knill1997errorcorrecting}.
In this representation, both $\vec{x}$ and $\vec{z}$ are binary vectors, where a $1$ in the $j$th position in $\vec{x}$ indicates that the Pauli operator contains an $X$ on qubit $j$, and a $1$ in the $j$th position in $\vec{z}$ indicates that the Pauli operator contains an $Z$ on qubit $j$ (so presence of both indices $Y$).
In other words, the Pauli operator that corresponds to $(\vec{x},\vec{z})$ is
\begin{equation}
\label{symplectic_pauli}
    P(\vec{x},\vec{z})=\bigotimes_{j=0}^{n-1}i^{x_jz_j}X^{x_j}Z^{z_j}
\end{equation}
(the $i$ in the above expression is the imaginary number, not the index of the Pauli operator).

Hence, this block encoding requires $2n$ auxiliary qubits for $n$ system qubits, so it triples the qubit requirement compared to the system alone.
In exchange, the gate costs are substantially lower, depending on the specific Hamiltonian to be encoded, so this block encoding option may be advantageous on near-term quantum computers whose noisy gates make shorter depth and fewer controls a higher priority than fewer qubits.
We can now express $U$ as
\begin{equation}
    U=\sum_{\vec{x},\vec{z}}|\vec{x},\vec{z}\rangle_a\langle\vec{x},\vec{z}|_a\otimes P(\vec{x},\vec{z}).
\end{equation}
Inserting \eqref{symplectic_pauli}, we can factor this as
\begin{equation}
\label{U_decomp}
    U=\prod_{j=0}^{n-1}i^{x_jz_j}\big(\text{ctrl}_{x_j}\text{-}X_j\big)\big(\text{ctrl}_{z_j}\text{-}Z_j\big),
\end{equation}
so $U$ can be implemented as only three layers of singly-controlled operations: a first layer that applies $Z_j$ to each qubit controlled on the corresponding auxiliary qubit encoding $z_j$, a second layer that applies $X_j$ to each qubit controlled on the corresponding auxiliary qubit encoding $x_j$, and a third layer that implements the two-qubit controlled phase
\begin{equation}
    \begin{pmatrix}
        1&0&0&0\\
        0&1&0&0\\
        0&0&1&0\\
        0&0&0&i
    \end{pmatrix}
\end{equation}
on each pair of auxiliary qubits $x_j,z_j$.
A circuit diagram for $U$ is shown in Fig. 5.

The complexity of preparing
\begin{equation}
    |G\rangle_a=\sum_{i=0}^{T-1}\sqrt{\alpha_i}|(\vec{x},\vec{z})_i\rangle_a,
\end{equation}
where $(\vec{x},\vec{z})_i$ is the binary representation of the $i$th term in the Hamiltonian, depends strongly on form of the Hamiltonian.
For each term $P_i$ in the Hamiltonian, the locality $l$ of $P_i$ (i.e., the number of qubits it acts upon) determines the Hamming weight of $(\vec{x},\vec{z})_i$, since each of $\vec{x}$ and $\vec{z}$ can have Hamming weight at most $l$.
Hence, one case in which this encoding may be advantageous is when the input Pauli Hamiltonian is local, i.e., $l$ is constant.

For example, the $l=2$ case includes nearest-neighbor spin models.
In this case, the Hamiltonian is composed of terms of the form
\begin{equation}
    \alpha_1\sigma_j\quad\text{or}\quad\alpha_2\sigma_j\sigma_k,
\end{equation}
where $\alpha_1,\alpha_2$ are coefficients and
\begin{equation}
    \sigma_j(x_j,z_j)=i^{x_jz_j}X_j^{x_j}Z_j^{z_j}
\end{equation}
denotes a single-qubit Pauli operator acting on qubit $j$.
The corresponding terms in $|G\rangle_a$ will be
\begin{equation}
\label{spin_terms}
\begin{split}
    &\sqrt{\alpha_1}|00...x_j...0,00...z_j...0\rangle,\\
    &\sqrt{\alpha_2}|00...x_j...x_k...0,00...z_j...z_k...0\rangle,
\end{split}
\end{equation}
in which the only potentially nonzero bits $x_j,x_k,z_j,z_k$ are in positions $j$ and $k$ in their respective bitstrings.

In order to prepare such a state, it is enough if we can prepare an arbitrary real superposition of computational basis states with Hamming weight one or two.
For weight one, we can accomplish this as follows:
\begin{equation}
\label{weight_one_seq}
\begin{split}
    |000...00\rangle\xrightarrow{X_0}&|100...00\rangle\\
    \xrightarrow{\text{PSWAP}_{01}(\theta_1)}&|010...00\rangle\\
    \xrightarrow{\text{PSWAP}_{12}(\theta_2)}&|001...00\rangle\\
    \vdots\quad&\\
    \xrightarrow{\text{PSWAP}_{n-2,n-1}(\theta_{n-1})}&|000...01\rangle,
\end{split}
\end{equation}
where after each step we only show the newest state in the superposition, and $\text{PSWAP}_{ij}(\theta)$ is defined as the two-qubit gate
\begin{equation}
\label{pswap_def}
    \text{PSWAP}_{ij}(\theta)=
    \begin{pmatrix}
        1&0&0&0\\
        0&\cos\theta&-\sin\theta&0\\
        0&\sin\theta&\cos\theta&0\\
        0&0&0&1
    \end{pmatrix}
\end{equation}
acting on qubits $i$ and $j$.
If we desired to obtain the coefficient $\beta_i\ge0$ for the bitstring with the one in the $i$th position, then the angles in \eqref{weight_one_seq} should be chosen as follows:
\begin{equation}
\label{rotation_angle}
\begin{split}
    &\cos\theta_1=\beta_1~\Rightarrow~\theta_1=\arccos\beta_1,\\
    &\sqrt{1-\beta_1^2}\cos\theta_2=\beta_2~\Rightarrow~\theta_2=\arccos\left(\frac{\beta_2}{\sqrt{1-\beta_1^2}}\right),\\
    &\qquad\vdots\\
    &\sqrt{1-\sum_{j=1}^{i-1}\beta_j^2}\cos\theta_i=\beta_i\\
    &\qquad\Rightarrow~\theta_i=\arccos\left(\frac{\beta_i}{\sqrt{1-\sum_{j=1}^{i-1}\beta_j^2}}\right).
\end{split}
\end{equation}
This procedure requires only linear qubit connectivity.

We can prepare an arbitrary real superposition of weight-two bitstrings as follows:
\begin{equation}
\label{weight_two_seq}
\begin{split}
    |000...00\rangle\xrightarrow{X_0X_1}&|1100000...0\rangle\\
    \xrightarrow{\text{C$_0$-PSWAP}_{12}(\theta_1)}&|1010000...0\rangle\\
    \xrightarrow{\text{C$_2$-PSWAP}_{01}(\theta_2)}&|0110000...0\rangle\\
    \xrightarrow{\text{C$_2$-PSWAP}_{13}(\theta_3)}&|0011000...0\rangle\\
    \xrightarrow{\text{C$_3$-PSWAP}_{21}(\theta_4)}&|0101000...0\rangle\\
    \xrightarrow{\text{C$_3$-PSWAP}_{10}(\theta_5)}&|1001000...0\rangle\\
    \xrightarrow{\text{C$_0$-PSWAP}_{34}(\theta_6)}&|1000100...0\rangle\\
    \xrightarrow{\text{C$_4$-PSWAP}_{01}(\theta_7)}&|0100100...0\rangle\\
    \xrightarrow{\text{C$_4$-PSWAP}_{12}(\theta_8)}&|0010100...0\rangle\\
    \xrightarrow{\text{C$_4$-PSWAP}_{23}(\theta_9)}&|0001100...0\rangle\\
    \xrightarrow{\text{C$_4$-PSWAP}_{35}(\theta_10)}&|0000110...0\rangle\\
    \xrightarrow{\text{C$_5$-PSWAP}_{43}(\theta_11)}&|0001010...0\rangle\\
    \xrightarrow{\text{C$_5$-PSWAP}_{32}(\theta_12)}&|0010010...0\rangle\\
    \xrightarrow{\text{C$_5$-PSWAP}_{21}(\theta_13)}&|0100010...0\rangle\\
    \vdots\quad&
\end{split}
\end{equation}
where again the state shown after each step is only the latest state to be added to the superposition and one can check that each controlled partial swap acts nontrivially on only the state obtained in the preceding step.
The controlled partial swap $\text{C$_k$-PSWAP}_{ij}(\theta)$ is simply the partial swap defined in \eqref{pswap_def}, controlled on qubit $k$.
Hence, if $\beta_i$ is the desired coefficient of the $i$th bitstring obtained in the order in \eqref{weight_two_seq}, the associated angle $\theta_i$ is again simply given by \eqref{rotation_angle}.
The above procedure requires arbitrary qubit connectivity, but this can be changed to linear connectivity by adding $O(n)$ SWAPs per step in \eqref{weight_two_seq} for $n$ qubits.

In general, to prepare states containing terms of the forms \eqref{spin_terms}, we need to be able to prepare an arbitrary real superposition of terms of Hamming weight one and two.
Let $\beta^{(1)}_i$ be the desired coefficients for the weight-one terms and $\beta^{(2)}_i$ be the desired coefficients for the weight-two terms.
To prepare the desired state, first prepare the weight-one part, only with $\beta^{(1)}_0$ replaced by
\begin{equation}
    \sqrt{(\beta^{(1)}_0)^2+\sum_{i}(\beta^{(2)}_i)^2}
\end{equation}
so that in the final state the term $|100...0\rangle$ has the above coefficient.
Then apply the partial CNOT
\begin{equation}
\label{partial_cnot}
    \text{CNOT}(\theta)=
    \begin{pmatrix}
        1&0&0&0\\
        0&1&0&0\\
        0&0&\cos\theta&\sin\theta\\
        0&0&\sin\theta&-\cos\theta
    \end{pmatrix}
\end{equation}
to qubits $0$ and $1$ with the angle
\begin{equation}
    \theta=\arccos\left(\frac{\beta^{(1)}_0}{\sqrt{(\beta^{(1)}_0)^2+\sum_{i}(\beta^{(2)}_i)^2}}\right).
\end{equation}
This will map the $|100...0\rangle$ term to
\begin{equation}
    \beta^{(1)}_0|100...0\rangle+\sqrt{\sum_{i}(\beta^{(2)}_i)^2}|110...0\rangle.
\end{equation}
Finally, apply the procedure \eqref{weight_two_seq} to prepare the weight-two part of the superposition, skipping the initial $X_0X_1$ step.
Since all of the C-PSWAPs do not change states of weight one, this will not disturb the already prepared weight-one part of the superposition, but will simply distribute the amplitude $\sqrt{\sum_{i}(\beta^{(2)}_i)^2}$ placed on $|110...0\rangle$ over the other weight-two states, as desired.

The above procedures suffice to prepare any superpositions of weight-one and weight-two bitstrings, which capture all terms of the form \eqref{spin_terms} corresponding to $X$, $Z$, $XX$, $ZZ$, or $XZ$ terms in the Hamiltonian.
To include terms containing $Y$s, we need to obtain ones in the same locations in the $\vec{x}$ and $\vec{z}$ registers.
This can be accomplished by doubly- and triply-controlled partial CNOTs, starting from states prepared as above and using a similar strategy to distribute amplitude.

The one remaining piece we require for the block encoding is to implement the signs of the terms.
Recall that the unitary $U$ in \eqref{unitary_app} is supposed to implement the signed Pauli operators so that the coefficients $\alpha_i$ in the Hamiltonian can be nonnegative.
However, our implementation of $U$ above simply applied each Pauli without a sign.
To include the signs, we can first assume without loss of generality that all single-qubit terms in the Hamiltonian are $Z$s with positive signs (this amounts to a choice of $|0\rangle$ and $|1\rangle$ states for each qubit).
Hence, we only have to implement signs on the interaction terms.

If the only interaction terms have the form $XX$, $ZZ$, or $XZ$, then the sign for each term can be implemented as a controlled-$Z$ on the pair of auxiliary qubits that are one in the binary representation of the term.
This requires at most $T$ controlled-$Z$ gates for $T$ terms.
If terms containing $Y$s are also present, then these phases must be triply-controlled, i.e., four-qubit gates instead.
A circuit diagram for $G$ is given in Fig. 6.

\subsection{Reflection about block encoding state}

The above completes the construction of the block encoding itself, but implementing qubitization as in \cref{chebyshev_lemma} also requires reflections around $|G\rangle_a$.
Our construction so far requires at most four-qubit gates, so we would prefer to avoid the strategy of simply inverting $G$, reflecting around $|0\rangle_a$, and then reapplying $G$, since this requires a phase controlled on the entire auxiliary register.
However, in the case of spin models, we can use the fact (pointed out in the proof of \cref{chebyshev_lemma}) that we do not in fact require a complete reflection around $|G\rangle_a$, but only within the subspace containing states of the auxiliary qubits that can be reached by applying $U$ to $|G\rangle_a$.

$U$ does not change the states of the auxiliary qubits directly, since it is only controlled by them, although this does mean that it entangles them with the system qubits.
In particular, $U$ does not change the Hamming weight of the auxiliary states.
To reflect around $|G\rangle_a$, we must invert $G$, i.e., apply $G^\dagger$ to the auxiliary qubits.
This operation can increase the Hamming weight of each of the registers $|\vec{x}\rangle$ and $|\vec{z}\rangle$ to at most three, which we can see as follows.
Prior to applying $G^\dagger$, the Hamming weights of $|\vec{x}\rangle$ and $|\vec{z}\rangle$ are each at most two.
The only operation that can increase the Hamming weight higher than this is the partial CNOT in \eqref{partial_cnot} applied to qubits 0 and 1 in each register to divide amplitude between the weight-one and weight-two terms.
When the inverse of this operation is applied to a state of weight two, if qubit 0 is in state $|1\rangle$ and qubit 1 is in state $|0\rangle$, while the other $|1\rangle$ is somewhere else, then one of the terms resulting from the partial CNOT will have weight three, with qubit 1 now contributing the third $|1\rangle$.
However, the only other operation to be applied to this state that does not conserve Hamming weight is the initial $X_0$ \eqref{weight_one_seq}, which decreases the weight of the state.
Of course, this $X_0$ may increase the weight of other computational basis states in the superposition, but none of those could have weight higher than two prior to its application, so overall it is impossible to end up with states of weight higher than three.

Given this, after inverting $G$, we only need to reflect around $|0\rangle_a$, the all-zeroes state of the auxiliary qubits, in the subspace spanned by states of Hamming weight up to three (for both registers $|\vec{x}\rangle$ and $|\vec{z}\rangle$).
This allows us to avoid reflecting by applying a phase controlled on the entire register.
Instead, we can simply use two additional auxiliary qubits to count the ones in $|\vec{x}\rangle$ and another two additional  auxiliary qubits to count the ones in $|\vec{z}\rangle$.
For example, to count the ones in $|\vec{x}\rangle$, we prepare its two auxiliary qubits in the state
\begin{equation}
    |w^{(x)}_1,w^{(x)}_0\rangle=|00\rangle,
\end{equation}
where $w^{(x)}_0$ and $w^{(x)}_1$ refer to the zeroth and first bits of the weight of $|\vec{x}\rangle$, respectively.
Then for each qubit in $|\vec{x}\rangle$, we apply Toffoli on $w^{(x)}_1$ controlled on $w^{(x)}_0$ and the current qubit in $|\vec{x}\rangle$, followed by CNOT on $w^{(x)}_0$ controlled on the current qubit in $|\vec{x}\rangle$.
Since there are at most three ones in $|\vec{x}\rangle$, this stores the number of ones in $|\vec{x}\rangle$ in $|w^{(x)}_1,w^{(x)}_0\rangle$ as a binary number.
Therefore, once we have done this for both $|\vec{x}\rangle$ and $|\vec{z}\rangle$, we can simply apply a $-1$ phase controlled on the resulting four auxiliary qubits being in state $|0000\rangle$.
After that, we uncompute the auxiliary qubits used for counting, then reapply $G$, and our reflection about $|G\rangle_a$ is complete.
A circuit diagram for $R$ is given in Fig. 9.

\subsection{Analysis}

Finally, we can analyze the costs of $U$, $G$, and $R$.
For simplicity, we will do this for the case where the Hamiltonian contains only $Z$, $ZZ$, and $XX$ terms, i.e., it is a Heisenberg XYZ model with the YY interaction set to zero.
For $U$, we require two main steps:
\begin{enumerate}
    \item We argued above that the cost prior to implementing the signs of the terms is three layers of singly-controlled operations as in \eqref{U_decomp}, for a total of $3n$ CNOTs and CPHASEs.
    
    \item To implement the signs of the terms, we require at most $T$ CZs for $T$ terms (one CZ per term with a negative sign).
\end{enumerate}
Hence the total cost to implement $U$ is
\begin{equation}
\label{U_cost}
    3n+T
\end{equation}
two-qubit gates.

For $G$, we also require two main steps, as outlined above:
\begin{enumerate}
    \item For each of $|\vec{x}\rangle$ and $|\vec{z}\rangle$, prepare the weight-one part of the superposition as in \eqref{weight_one_seq}. This requires a single NOT followed by $n-1$ PSWAPs for each, so in total requires $2$ NOTs and $2(n-1)$ PSWAPs.
    
    \item For each of $|\vec{x}\rangle$ and $|\vec{z}\rangle$, prepare the weight-one part of the superposition as in \eqref{weight_two_seq}, replacing the first step (the double NOT) with the partial CNOT as in \eqref{partial_cnot}. This requires the partial CNOT followed by
    \begin{equation}
        \sum_{i=2}^{n-2}i=\frac{1}{2}n^2-\frac{3}{2}n
    \end{equation}
    C-PSWAPs, so in total requires $2$ PCNOTs and $n^2-3n$ C-PSWAPs.
\end{enumerate}
We obtain the total cost for $G$ by adding up the above gate counts, but first we note that the three qubit gates, which are all C-PSWAPS, can be implemented in place via four two-qubit gates each (this construction is inspired by the contruction of controlled time-evolutions for certain Hamiltonians in~\cite{dong2022groundstate}).
For
\begin{equation}
    \hat\mu_{jk}\coloneqq
    \begin{pmatrix}
        0&0&0&0\\
        0&0&i&0\\
        0&-i&0&0\\
        0&0&0&0
    \end{pmatrix},
\end{equation}
where $j,k$ denote the qubits acted upon, we can express a partial swap as
\begin{equation}
    \text{PSWAP}_{jk}(\theta)=e^{i\theta\hat\mu_{jk}}.
\end{equation}
We also have
\begin{equation}
    \hat\mu_{jk}Z_j=-Z_j\hat\mu_{jk},
\end{equation}
so
\begin{equation}
    \text{C$_j$-PSWAP$_{kl}$}(\theta)=e^{i\theta\hat\mu_{kl}/2}\cdot C_jZ_k\cdot e^{-i\theta\hat\mu_{kl}/2}\cdot C_jZ_k,
\end{equation}
a product of four two-qubit gates, as desired.
Using this, the total cost for $G$ comes to
\begin{equation}
\label{G_cost}
\begin{split}
    &4n^2-10n\quad\text{two-qubit},\\
    &2\quad\text{single-qubit}
\end{split}
\end{equation}
gates (for $n\ge3$).

Finally, implementing $R$ requires an application of $G^\dagger$ and an application of $G$, with the reflection about $|0\rangle_a$ in between.
The cost of the latter is $4n$ Toffoli gates and $4n$ CNOTs (to compute and uncompute the weight of the state), plus the four-qubit CPHASE.
As in the implementation of $G$, we can replace each three qubit gate (Toffolis in this case) with four two-qubit gates, this time using
\begin{equation}
    \hat\nu_{jk}\coloneqq
    \begin{pmatrix}
        0&0&0&0\\
        0&0&0&0\\
        0&0&0&1\\
        0&0&1&0
    \end{pmatrix}
\end{equation}
to generate a CNOT as
\begin{equation}
    \text{CNOT}_{jk}=e^{i\pi\hat\nu_{jk}/2}S_j^\dagger,
\end{equation}
where
\begin{equation}
    S_j\coloneqq\begin{pmatrix}1&0\\0&i\end{pmatrix}
\end{equation}
acts on the controlling qubit.
As above, we also have
\begin{equation}
    \hat\nu_{jk}Z_j=-Z_j\hat\nu_{jk},
\end{equation}
Therefore,
\begin{equation}
\label{toffoli_decomp}
    \text{Toffoli}_{jkl}=e^{i\pi\hat\nu_{kl}/4}\cdot C_jZ_k\cdot e^{-i\pi\hat\nu_{kl}/4}\cdot C_jZ_k\cdot C_jS_k^\dagger,
\end{equation}
and the $C_jS_k^\dagger$ can be absorbed into the second $C_jZ_k$, yielding a sequence of four two-qubit gates, as desired.
The four-qubit CPHASE can be implemented in place via fourteen two-qubit gates as follows.
Note that for
\begin{equation}
    \hat\omega_{jk}\coloneqq
    \begin{pmatrix}
        0&0&0&0\\
        0&0&0&0\\
        0&0&1&0\\
        0&0&0&-1
    \end{pmatrix},
\end{equation}
we have
\begin{equation}
    C_jZ_k=e^{i\pi\hat\omega_{jk}/2}S_j^\dagger,
\end{equation}
and
\begin{equation}
    \hat\omega_{jk}X_j=-X_j\hat\omega_{jk}.
\end{equation}
Hence,
\begin{equation}
\begin{split}
    &CCCZ_{jklm}\\
    &=e^{i\pi\hat\omega_{lm}/4}\cdot\text{Toffoli}_{jkl}\cdot e^{-i\pi\hat\omega_{lm}/4}\cdot\text{Toffoli}_{jkl}\cdot CCS_{jkl}^\dagger,
\end{split}
\end{equation}
where the Toffolis can be implemented via \eqref{toffoli_decomp} and the $CCS^\dagger$ can be implemented similarly, so that the total cost for the four-qubit CPHASE is fourteen two-qubit gates.
Hence the total cost of $R$ is twice the cost of $G$, plus $4n$ Toffoli gates and $4n$ CNOTs, plus the four-qubit CPHASE, which comes to a total of
\begin{equation}
\label{R_cost}
\begin{split}
    &8n^2+14\quad\text{two-qubit},\\
    &4\quad\text{single-qubit}
\end{split}
\end{equation}
gates.
Implementing $R$ also requires six additional auxiliary qubits in total.

\clearpage

\centering
\includegraphics[width=\textwidth]{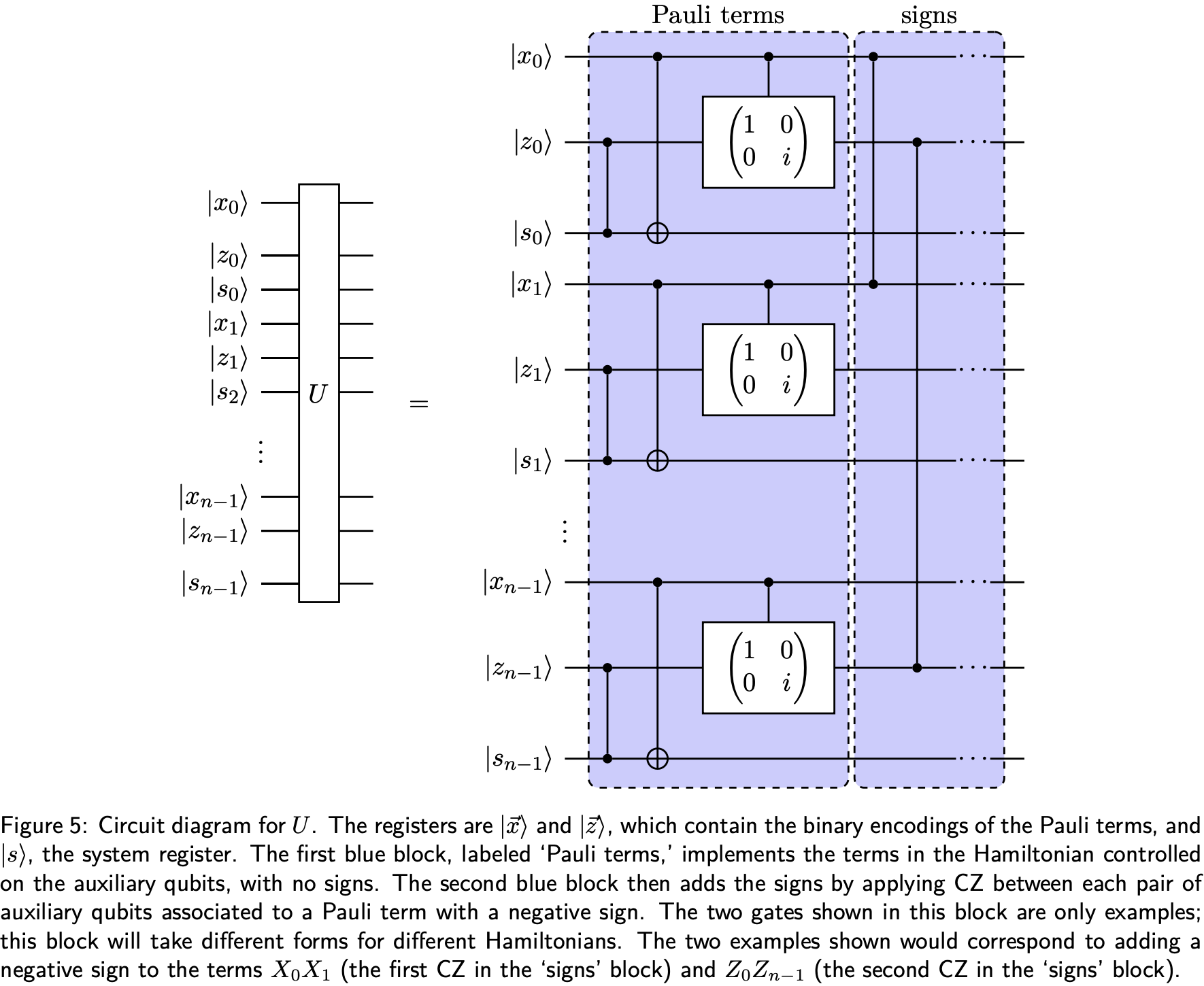}

\clearpage

\centering
\includegraphics[width=\textwidth]{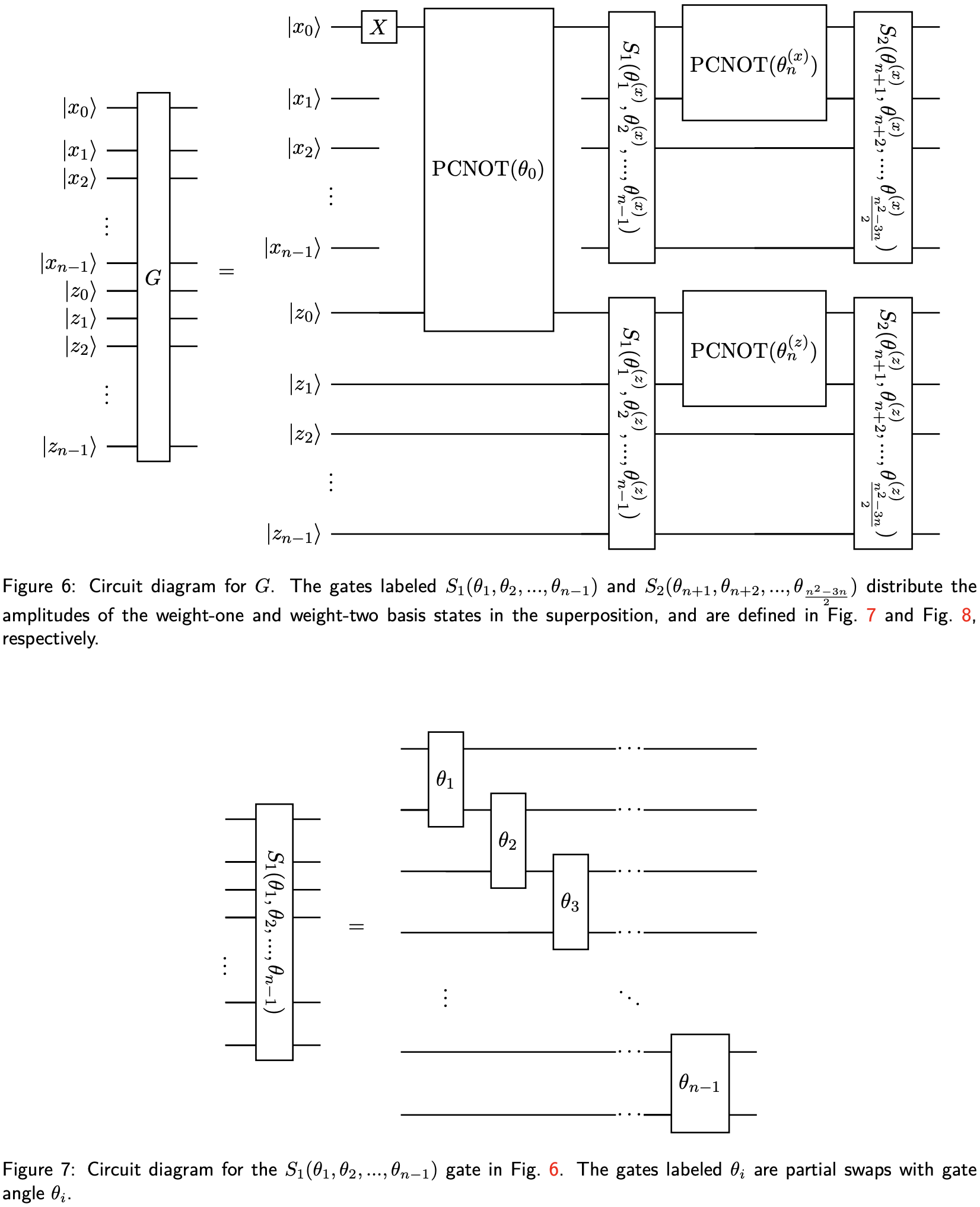}

\clearpage

\centering
\includegraphics[width=\textwidth]{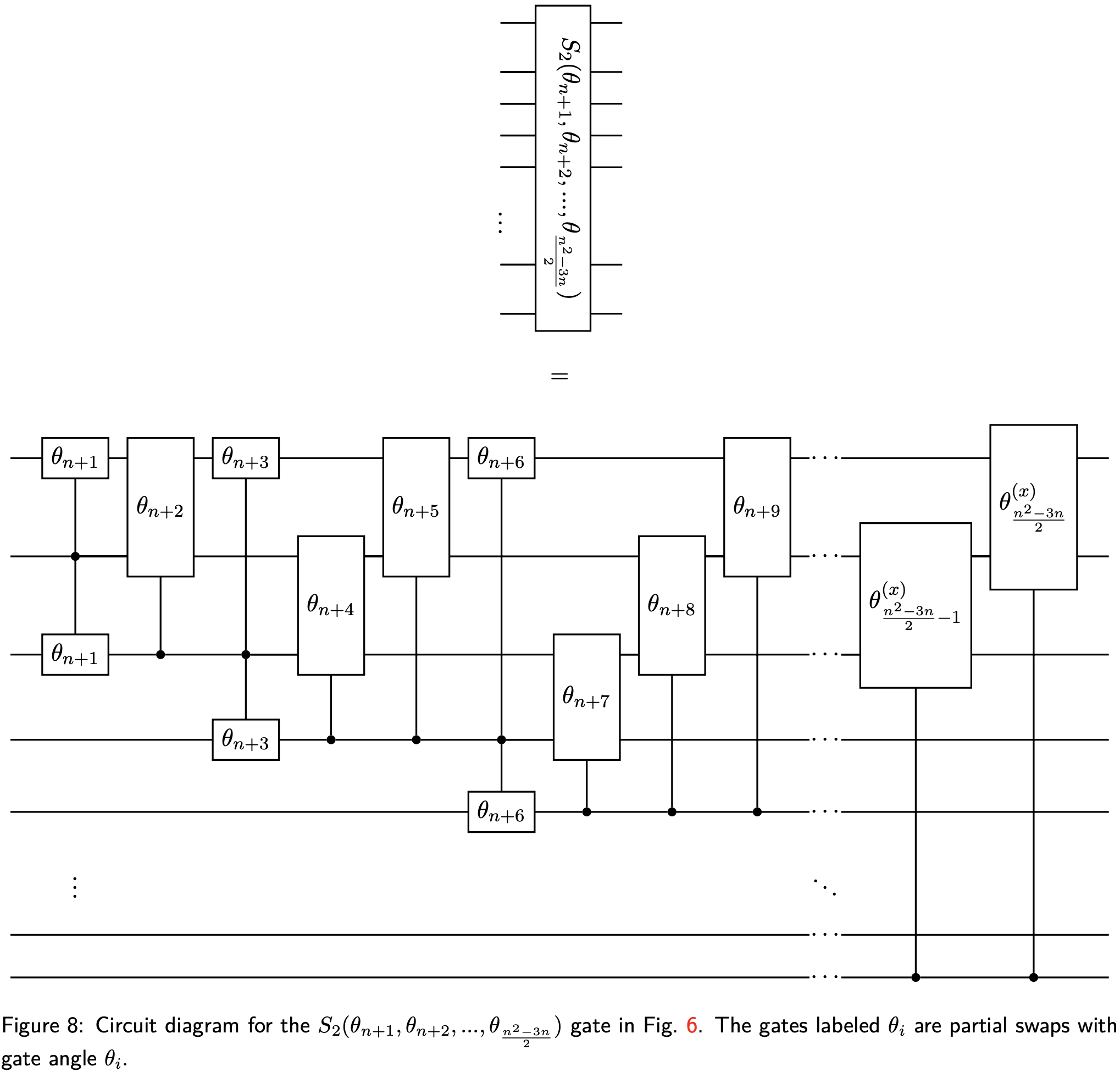}

\clearpage

\centering
\includegraphics[width=\textwidth]{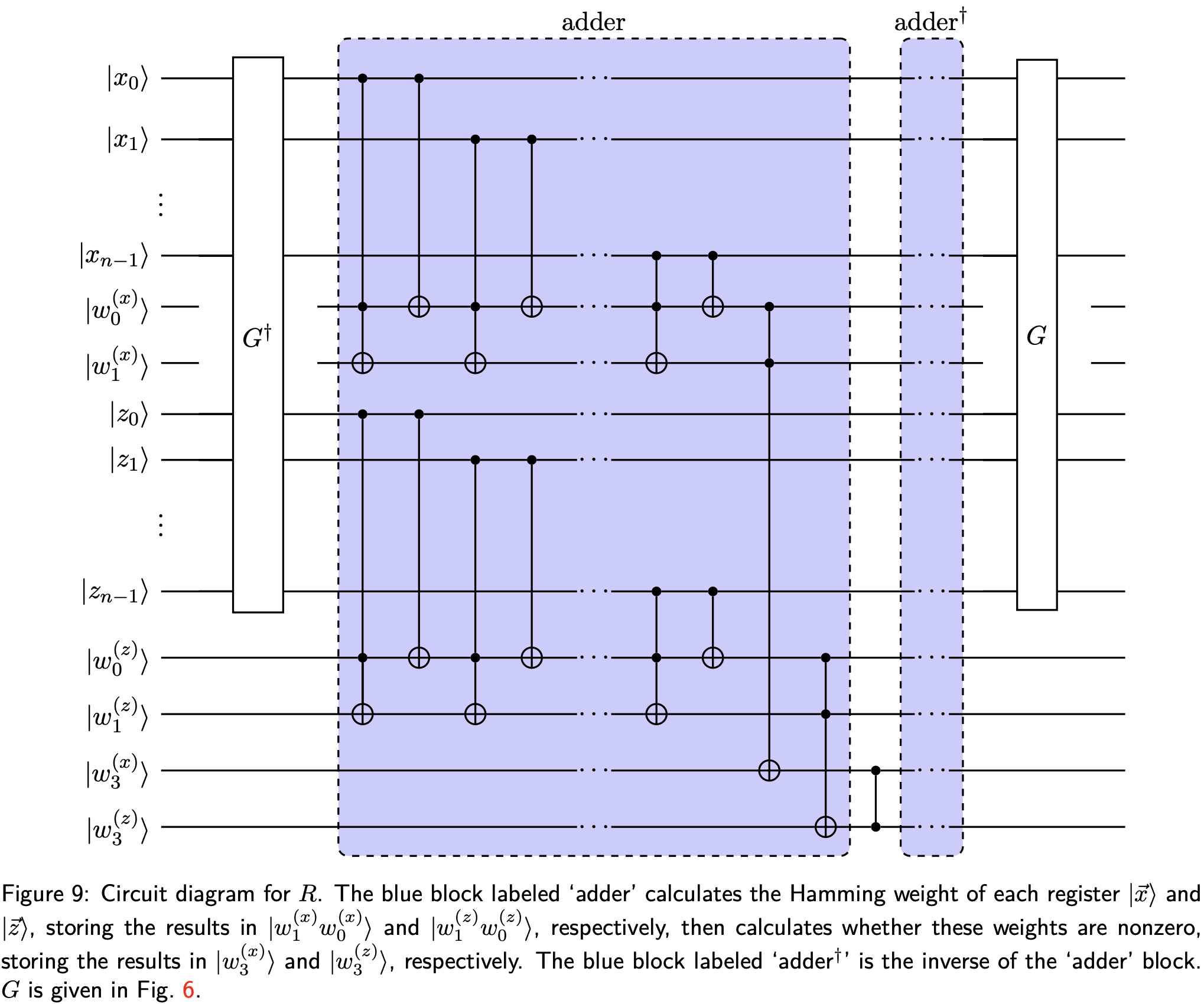}

\end{document}